\documentclass[12pt]{amsart}
\usepackage[margin=1.2in]{geometry}

\usepackage{amsmath,amssymb,amsfonts,amsthm,bbm,graphicx,color,tikz}
\usepackage{pgf}
\usepackage{mathrsfs}

\usepackage{hyperref}
\usepackage{cleveref}
\usepackage[normalem]{ulem}
\usepackage[english,french]{babel}

\usepackage[utf8]{inputenc}

\usetikzlibrary{decorations.markings}
\usetikzlibrary{arrows}
\usetikzlibrary{matrix}
\newtheorem{theorem}{Theorem}

\newtheorem{corollary}[theorem]{Corollary}

\allowdisplaybreaks

\newtheorem{lemma}[theorem]{Lemma}

\newtheorem{proposition}[theorem]{Proposition}
\newtheorem{remark}[theorem]{Remark}

\DeclareMathOperator{\Real}{Re}
\DeclareMathOperator{\Imaginary}{Im}
\renewcommand{\Re}{\Real}
\renewcommand{\Im}{\Imaginary}
\newcommand{\C}{{\mathbb C}}

\newcommand{\R}{{\mathbb R}}

\newcommand{\Z}{{\mathbb Z}}
\newcommand{\G}{{\mathcal G}}

\renewcommand{\H}{{\mathcal H}}

\newcommand{\D}{{\mathbb D}}

\newcommand{\be}{\begin{equation}}
\newcommand{\ee}{\end{equation}}
\newcommand{\GD}{\G_{\text D}}

\definecolor{calpolypomonagreen}{rgb}{0.12, 0.65, 0.17}

\title[Dimers and circle patterns]{Dimers and circle patterns}
\author{Richard Kenyon, Wai Yeung Lam, Sanjay Ramassamy, Marianna Russkikh}

\address{Richard Kenyon\\
	Mathematics Department \\Yale University\\ New Haven\\ CT 06520}

\email{richard.kenyon at yale.edu}

\address{Wai Yeung Lam\\
	Beijing Institute of Mathematical Sciences and Applications\\ Huairou District\\ Beijing\\ P.R. China}

\email{lam at bimsa.cn}

\address{Sanjay Ramassamy\\ 
	D\'epartement de math\'ematiques et applications \\ \'Ecole normale sup\'erieure \\ CNRS \\ PSL University \\ 45 rue d'Ulm \\ 75 005 Paris, France}

\email{sanjay.ramassamy at ipht.fr}

\address{Marianna Russkikh\\
	Section de math\'ematiques \\ Universit\'e de Gen\`eve \\ 2-4 rue du Li\`evre, case postale 64 \\ 1211 Gen\`eve 4, Switzerland}

\email{marianna.russkikh at unige.ch}

\begin{document}

\maketitle
\selectlanguage{english}
\begin{abstract} 
We establish a correspondence between the dimer model on a bipartite graph and a circle pattern with the combinatorics of that graph, which holds for graphs that are either planar or embedded on the torus. The set of positive face weights on the graph gives a set of global coordinates on the space of circle patterns with embedded dual. Under this correspondence, which extends the previously known isoradial case, the urban renewal (the local move for dimer models) is equivalent to the Miquel move (the local move for circle patterns). As a consequence, we show that \emph{Miquel dynamics} on circle patterns is a discrete integrable system governed by the octahedron recurrence. As special cases of these circle pattern embeddings, we recover harmonic embeddings for resistor networks and s-embeddings for the Ising model.
\end{abstract} {\selectlanguage{french}
\begin{abstract}
	Nous \'etablissons une correspondance entre le modèle de dimères sur un graphe biparti et un agencement de cercles avec la combinatoire de ce graphe, valable pour des graphes plongés sur le plan ou sur le tore. Les poids positifs sur les faces du graphe fournissent des coordonnées globales sur l'espace des agencements de cercles dont le dual est plongé. Via cette correspondance, qui étend le cas isoradial découvert précédemment, le renouvellement urbain (mouvement local pour les modèles de dimères) est équivalent au mouvement de Miquel (mouvement local pour les agencements de cercles). Il en découle que la dynamique de Miquel sur les agencements de cercles est un système intégrable discret gouverné par la récurrence de l'octaèdre. Comme cas particuliers de ces plongements comme agencements de cercles, on retrouve les plongements harmoniques pour les réseaux de résistances et les s-plongements pour le modèle d'Ising.
\end{abstract}}

\tableofcontents
\section{Introduction}

The \emph{bipartite planar dimer model} is the study of random perfect matchings (``dimer coverings") of a bipartite planar graph. The dimer model is a classical statistical mechanics model, and can be analyzed using determinantal methods: partition functions and correlation
kernels are computed by determinants of associated matrices defined from the weighted graph \cite{Kenyon.localstats}. 
Several other
two-dimensional models of statistical mechanics, including the Ising model and the spanning tree model, 
can be regarded as special cases of the dimer model by subdividing the underlying graph \cite{Fisher1966,Temperley1974,Dubedat2011, Kenyon2000}. Natural parameters for the dimer model, defining the underlying probability measure, 
are \emph{face weights}, which are positive real parameters on the bounded faces of the graph~\cite{GK2013}.

A \emph{circle pattern} is a realization of a graph in $\C$ with cyclic faces, i.e. where all vertices on a face lie on a circle.
Circle patterns are central objects in discrete differential geometry, related to (hyperbolic) polyhedra, Teichm\"uller space, and discrete conformal geometry. For example, following original ideas of William Thurston, two circle patterns with the same intersection angles are considered discretely conformally equivalent, see e.g.~\cite{Bobenko2004}.

In \cite{Kenyon2002} a relation was found between a special subset of dimer models,
called \emph{critical dimer models}, and \emph{isoradial circle patterns}, i.e. circle patterns in which all the circles have the same radius. The partition function
and various probabilistic quantities were related to the underlying 3D hyperbolic geometry. At that time there was no clear relation between general dimer models and general circle patterns and this question was raised again in~\cite{Bobenko2010} and~\cite{Ramassamy2017}.

The main purpose of this paper is to answer this question, establishing a correspondence between face-weighted bipartite planar graphs and circle patterns, which generalizes the isoradial case. This correspondence is formulated for two classes of planar graphs, finite graphs and infinite bi-periodic graphs. Under this correspondence, dimer face weights correspond to biratios of distances between circle centers. A major feature of this correspondence is to identify the \emph{spider move} (also known as \emph{urban renewal},
or cluster mutation), which is a local move for the dimer model, to an application of Miquel's six-circles theorem to the underlying circle pattern. This establishes a new connection between circle patterns and cluster algebras.

The circle patterns arising under this correspondence are those with a bipartite graph and with an \emph{embedded dual}, where the dual graph is the graph of circle centers.
Having embedded dual does not imply that the primal pattern is embedded, although the set of circle patterns with embedded dual includes all embedded circle patterns in which each face contains its circumcenter. Centers of bipartite circle patterns arise in various places. They coincide with the crease patterns of origami that are locally flat-foldable \cite{Hull2002}. In discrete differential geometry, they are called conical meshes \cite{Pottmann2008,Muller2015} and related to discrete minimal surfaces \cite{Lam2016}. 
Circle center embeddings are also considered in~\cite{CLR} under the name of t-embeddings with an emphasis on the convergence of discrete holomorphic functions to continuous ones in the small mesh size limit, i.e., when the circle radii tend to~$0$.

This correspondence between dimer models of statistical mechanics and circle patterns from discrete differential geometry should allow to transfer results from one field to the other. As a first application of this correspondence, we show that \emph{Miquel dynamics}, a discrete-time dynamical system for periodic circle patterns introduced in \cite{Ramassamy2017} and also studied in \cite{Glutsyuk2018}, is a discrete integrable system governed by the octahedron recurrence, answering a conjecture made in \cite{Ramassamy2017}. In the dimer model our natural parameters are the positive face weights, which correspond on the level of circle patterns to embedded circle centers. However our results apply to general real face weights and non-embedded circle centers as well; in particular Miquel dynamics is algebraic in nature and the signs of the weights do not matter.

A central question in 2D statistical mechanics is to find embeddings of planar graphs which are adapted to a model and at the same time universal, i.e. with a definition valid for any planar graph, see e.g. \cite{Beffara08}. While the definition of a statistical mechanics model on a planar graph (e.g. random walk, dimer model, Ising model) does not depend on the embedding of the graph, stating and proving scaling limit results to conformally covariant objects such as Brownian motion or SLE curves requires one to pick an appropriate embedding for the graph. \emph{Harmonic embeddings} (also known as Tutte embeddings) provide such adapted embeddings for resistor networks and random walks (see e.g. \cite {Biskup11}). The \emph{s-embeddings} recently introduced by Chelkak \cite{Chelkak2017} (see also \cite{Lis2017}) are embeddings adapted to the Ising model. 

Our main result is that circle center embeddings are the right universal framework to study the planar bipartite dimer model. A first indication of this is the aforementioned compatibility between the local moves for the dimer model and for circle patterns. A second indication is that both resistor networks and the Ising model on planar graphs can be seen as special cases of the bipartite dimer model \cite{Kenyon2000, Dubedat2011} and we show in this article that both harmonic embeddings and s-embeddings arise as special cases of circle center embeddings.

There is an intriguing algebraic similarity between the dimer model and Teichm\"uller theory: The face weights describing the dimer model \cite{GK2013} and the shear coordinates for Teichm\"uller space \cite{Fock2006} both behave like X-variables from cluster algebras. The relation between Teichm\"uller theory and circle patterns together with the correspondence between dimers and circle patterns could help to shed light on this similarity.

During the completion of this work, a preprint by Affolter \cite{Affolter2018} appeared, which shows how to go from circle patterns to dimers and observes that the Miquel move is governed by the central relation. Affolter notes that there is some information missing to recover the circle pattern from the $X$ variables. We provide here a complete picture, both in the planar and torus cases.

\subsection*{Organization of the paper}
In Section~\ref{sec:bipartiteembeddings}, we introduce circle center embeddings associated with bipartite graphs with positive face weights in the planar case. 
Section~\ref{sectionbiperiodic} is devoted to circle center embeddings in the torus case.
In Section~\ref{sec:Miquel} we show the equivalence between the 
spider move/urban renewal for the bipartite dimer model and the central move coming from Miquel's theorem for circle patterns. In particular this
gives a cluster algebra structure underlying Miquel dynamics. Section~\ref{sec:networks} is devoted to translating into circle geometry the generalized Temperley bijection between resistor networks and dimer models. In 
Section~\ref{sec:sembeddings} we show that the s-embeddings for the Ising model arise as a special case of circle center embeddings.

\section{Background on dimers and the Kasteleyn matrix}

For general background on the dimer model, see \cite{Kenyon.lectures}.
A \emph{dimer cover}, or \emph{perfect matching}, of a graph $\G$ is a set of edges with the property that every vertex
is contained in exactly one edge of the set. We assume our graphs are finite, connected and embeddable either on the plane or on the torus. 
A graph is \emph{nondegenerate} (for the dimer model) if it has dimer covers,
and each edge occurs in some dimer cover.

If $\nu:E(\G)\to\R_{>0}$ is a positive weight function on edges of $\G$ , we associate a
weight $\nu(m)=\prod_{e\in m}\nu(e)$ to a dimer cover which is the product of its edge weights. We can also associate 
to this data a probability measure $\mu$ on the set
$M$ of dimer covers, giving a dimer cover $m$ a probability $\frac{1}{Z}\nu(m)$, where $Z=\sum_{m\in M}\nu(m)$ 
is a normalizing constant,
called the \emph{partition function}. 

Two weight functions $\nu_1,\nu_2$
are said to be \emph{gauge equivalent} if
there is a function $F:V(\G)\to\R$ such that for any edge $vv'$, $\nu_1(vv') = F(v)F(v')\nu_2(vv').$
Gauge equivalent weights define the same probability measure $\mu$.
For a planar bipartite graph, two weight functions are gauge equivalent if and only if their \emph{face weights} are equal,
where the face weight of a face with vertices $w_1,b_1,\dots,w_k,b_k$ is the ``alternating product'' of its edge weights,
\be\label{Xdef}X=\frac{\nu(w_1b_1)\dotsm\nu(w_kb_k)}{\nu(b_1w_2)\dotsm\nu(b_kw_1)}.\ee

If $\G$ is a planar bipartite graph which has dimer covers, a \emph{Kasteleyn matrix} is a signed, weighted 
adjacency matrix, with rows indexed by the
 white vertices and columns indexed by the black vertices, with $K_{wb}=0$ if $w$ and $b$ are not adjacent, 
and $K_{wb} = \pm\nu(wb)$ otherwise, where the signs are chosen so that the product of signs around a face
is~$(-1)^{k+1}$ for a face of degree $2k$. Kasteleyn \cite{Kast} showed that the determinant of a Kasteleyn matrix is the weighted
sum of dimer covers: 
$$|\det K| = Z = \sum_{m\in M}\nu(m).$$

Different choices of signs satisfying the Kasteleyn condition correspond to multiplying $K$ on the right and/or left by
diagonal matrices with $\pm1$ on the diagonals. Different choices of gauge correspond to 
multiplying $K$ on the right and left by
diagonal matrices with positive diagonal entries (see e.g. \cite{GK2013}). We call two matrices  \emph{gauge equivalent} if they are related by these two operations: multiplication on the right and left by diagonal matrices with real nonzero diagonal entries.
Note that in terms of any (gauge equivalent) Kasteleyn matrix we can recover the face weights via the formula
\be\label{XdefK}X=(-1)^{k+1}\frac{K_{w_1b_1} K_{w_2b_2}\dotsm K_{w_kb_k}}{K_{w_1b_k} K_{w_2b_1}\dotsm K_{w_kb_{k-1}}}.\ee

In some circumstances it is convenient to take complex signs $e^{i\theta}$ in the Kasteleyn matrix, rather than just  $\pm1$;
in that case the required condition on the signs is that the quantity $X$ in (\ref{XdefK}) is positive, see \cite{Percus}.
This generalization will be used below.

Certain elementary transformations of $\G$ preserve the measure $\mu$; see Figure \ref{elemtransfs}.
\begin{figure}
\begin{center}
\includegraphics[width=3.5in]{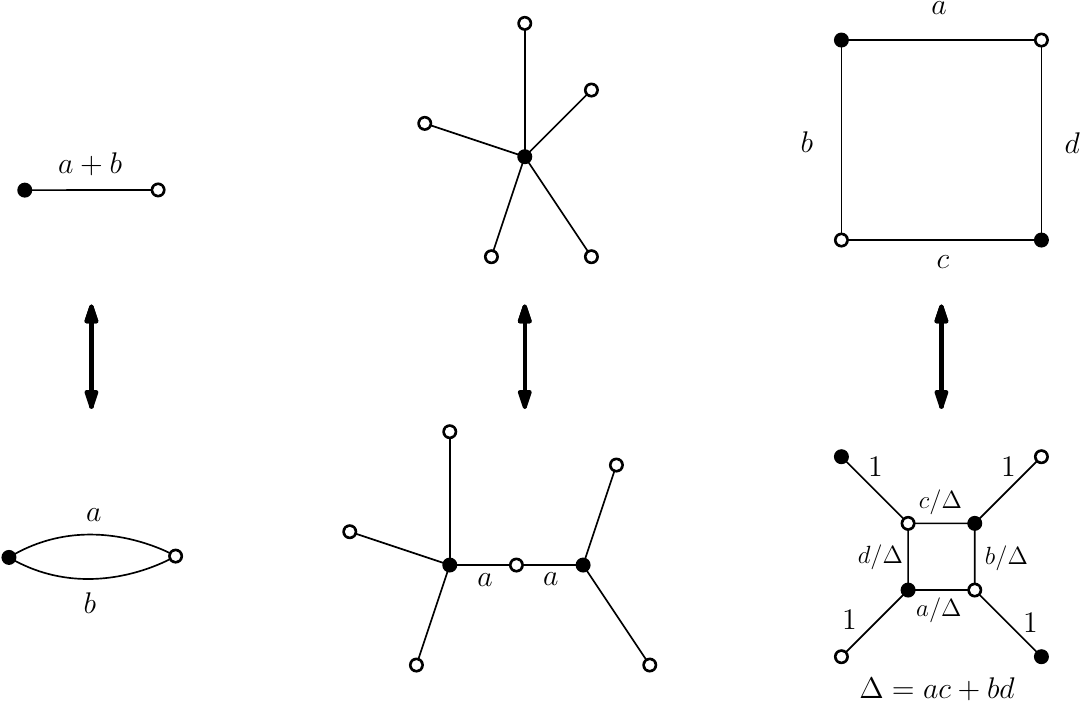}
\end{center}
\caption{\label{elemtransfs}Elementary transformations preserving the dimer measure~$\mu$. {\bf Left:} Replacing parallel edges
with weights $a,b$ by a single edge with weight $a+b$; {\bf Center:} contracting a degree $2$ vertex whose edges have equal weights; {\bf Right:} the spider move, with weights transformed as indicated.}
\end{figure}

\section{Bipartite graphs and circle patterns}
\label{sec:bipartiteembeddings}

In this section, we establish a correspondence between planar bipartite graphs with positive face weights and circle patterns with embedded dual. The construction can be extended to general real weights, but in general the dual will not be embedded unless weights are positive  (Theorem~\ref{thm:planardimerstocircles} below). 

The correspondence holds for several different types of boundary conditions. Although in the intermediate steps
we discuss somewhat general boundary conditions, for the final result (Theorem \ref{thm:planardimerstocircles}) we will 
need to consider only the (simplest) case of a circle pattern with outer face of degree $4$ (which is necessarily cyclic).

\subsection{Centers of circle patterns}

Let $\G$ be a finite connected embedded bipartite planar graph. Let $\hat\G$ be obtained from $\G$ by adding a vertex $v_\infty$ connected to all vertices on $\G$'s outer face.
Let $\G^*$ be the planar dual of $\hat\G$,  where~$v_\infty$ corresponds to the outer face of~$\G^*$. We call the vertices of $\G^*$ on its outer face the \emph{outer dual vertices}.
There is one outer dual vertex for every edge on the outer face of $\G$. We 
refer to $\G^*$ as the \emph{augmented dual} of $\G$ to distinguish it from the usual dual. 

Suppose $c: V(\G) \to \mathbb{C}$ is an embedding of 
$\G$ with cyclic faces (i.e. for all vertices $v$ of a single face $f$, all points $c(v)$ lie on a single circle), except perhaps the outer face, which we assume to be convex. Assume also that each bounded
face contains its circumcenter. 
  
The circumcenters then form an embedding $\widetilde\phi:V(\G^*) \to \mathbb{C}$ 
of the graph $\G^*$, except for the outer dual vertices. For each outer dual vertex $f$ of $\G^*$ define 
$\widetilde\phi(f)$
to be a point on the perpendicular bisector of the corresponding edge of $\G$, and external to the convex hull of $\G$.
We can think of $\widetilde\phi(f)$ as the center of a circle passing 
through the two vertices of the corresponding outer edge of $\G$.

Since each dual edge connects the centers of two circles with the corresponding primal edge as a common chord, each 
dual edge is a perpendicular bisector of the primal edge. 

Recalling that $\G$ is bipartite, note that 
the alternating sum of angles around every non-outer
vertex of $\widetilde\phi(\G^*)$ is zero. 
Moreover note that the faces of the augmented dual graph $\widetilde\phi(\G^*)$ are convex: we have a 
\emph{convex embedding}, that is, an embedding with convex faces, of $\G^*$.

The following proposition provides a partial converse to this construction.

\begin{proposition}\label{angles}
Suppose $\G$
is a bipartite planar graph and $\widetilde\phi:V(\G^*) \to \mathbb{C}$ is a convex embedding of $\G^*$. 
Then there exists a circle pattern $c:V(\G)  \to \mathbb{C}$ with $\widetilde\phi$ as centers if and only if the alternating sum of angles
around every non-outer dual vertex is zero.
\end{proposition}
 
Note that we do not conclude that $c(V(\G))$ is an embedding, only a realization with the property that vertices on each face lie on a circle. It seems difficult to give conditions under which $c( V(\G))$  will be an embedding, although the space of circle pattern embeddings in which each face contains the circumcenter
is an open subset of our space of realizations. Furthermore, if a circle pattern exists, it is not unique. Indeed there is a two-parameter family of circle patterns with $\widetilde\phi$  as centers, 
which depends on the position of an initial vertex as shown in the proof.

\begin{proof}
It remains to show that given such an embedding $\widetilde\phi$, there is a circle pattern 
with~$\widetilde\phi$ as centers. We construct such a circle pattern~$c$  as follows. Pick a vertex~$v_i$ and assign the vertex to some arbitrary point~$c(v_i)$ in the plane. We then define $c(v_j)$ for a neighboring vertex $v_j$ in such a way that $c(v_j)$ is the image of $c(v_i)$  under reflection across the line connecting images of the neighboring dual vertices $f$ and $f'$. 
Because of the angle condition, iteratively defining the~$c$ value around a face will return to the initial value. Hence the map~$c$  is well defined and independent of the path chosen. Note that $|c(v_i)\widetilde\phi(f)| = |c(v_j)\widetilde\phi(f)|$ and $|c(v_i)\widetilde\phi(f')| = |c(v_j)\widetilde\phi(f')|$, therefore the faces under the map $c$ are cyclic with centers at $\widetilde\phi$.
\end{proof}

\subsection{From circle patterns to face weights}
\label{circlestoX}

Suppose $\G$
is a bipartite planar graph and we have an \emph{embedded} circle pattern $c: V(\G) \to\C$, in which each bounded face contains its circumcenter, with outer face convex but not necessarily cyclic. Let $ \widetilde\phi:V(\G^*) \to \C$ be the circle centers (defining $\widetilde\phi$ on outer dual vertices as above). 

Now define a function $\omega(w,b)= \widetilde\phi(f_l) - \widetilde\phi(f_r)$ where  $f_l, f_r$ 
denote the left and the right face of the edge $wb$ 
oriented from $w$ to $b$. 
Define a matrix $K$ with rows indexing the white vertices and columns indexing the black 
vertices by $K_{wb}=\omega(w,b)$. We claim that $K$ 
is a Kasteleyn matrix (with complex signs).
To see this, suppose a bounded face $f$ with center $u=\widetilde\phi(f)$ 
has vertices $c(w_1),c(b_1),\dots, c(w_k),c(b_k)$ in counterclockwise order. 
We denote the centers of the neighboring faces as $u_1,u_2,\dots,u_{2k}$, where $u_1$ is the left face of $w_1b_k$ and $u_2$ is the right face of $w_1b_1$. Then
	\[
	\frac{K_{w_1b_1}K_{w_2b_2}\dotsm K_{w_kb_k}}{K_{w_1b_k}K_{w_2b_1}\dotsm K_{w_kb_{k-1}}}= \frac{(u-u_2)(u-u_4)\dotsm(u-u_{2k})}{(u_1-u)(u_3-u)\dotsm(u_{2k-1}-u)}.
	\]
The angle condition is equivalent to saying that the face weight
\be\label{eq:realcenter}
	X_{f} =  (-1)^{k+1}\frac{(u-u_2)(u-u_4)\dotsm(u-u_{2k})}{(u_1-u)(u_3-u)\dotsm(u_{2k-1}-u)}
\ee
is positive and $K$ is a Kasteleyn matrix. This associates a positive-face-weighted bipartite planar graph to a circle pattern.

We claim, additionally, that if the outer face of $c(\G)$ is cyclic 
(which will be the case if it has degree $4$, see below) 
the graph $\G$ has dimer covers and is nondegenerate (each edge is an element of at least one dimer cover).
The existence of dimer covers follows if we can find a \emph{fractional dimer cover} (fractional matching),
that is, an element of $[0,1]^E$ summing to $1$ at each vertex: recall that the set of dimer covers of a graph is the set
of vertices of the polytope of fractional dimer covers \cite{LovaszPlummer}. To find a fractional dimer cover,
associate to each edge $wb$ the quantity $\frac{\theta_{wb}}{2\pi}$ where $\theta_{wb}$ is the angle at $w$ 
(or $b$, they are the same) 
of the quad $(w,f,b,f')$ whose vertices are the vertices $w,b$ 
and the two dual vertices $f,f'$ of that edge.
In the case that one of these dual vertices
is an outer dual vertex, define the angle
$\theta_{wb}$ at $w$ instead as follows. Let $u_{f_o}$ be the circumcenter of the outer face, and $u$ be the other dual 
vertex of the edge $wb$. Define $\theta_{wb}=\angle uwb + \angle bwv$ where~$v$ is a point lying past $w$ on the ray from $u_{f_o}$ through~$w$.
Note that the angle obtained at~$b$ by the analogous method will equal that at $w$. 
Since $\sum_{b}\theta_{wb}=2\pi = \sum_{w}\theta_{wb}$, this defines a fractional dimer cover.
The nondegeneracy follows from the fact that the fractional matching is nonzero on each edge.

\subsection{Coulomb gauge for finite planar graphs with outer face of degree $4$}

In this section let $\G$ be a face-weighted bipartite planar graph, which is nondegenerate, and
has outer face of degree $4$. Let $Q$ be a convex quadrilateral. We 
construct a circle pattern for $\G$ with $\G^*$ embedded in $Q$. See Figure \ref{quadexample} for an example.
Our inductive construction will in principle work for graphs with outer face of higher degree,
but the initial step of the induction proof is more complicated and is not something we can currently handle.

\begin{figure}[htbp]
\center{\includegraphics[width=1.6in]{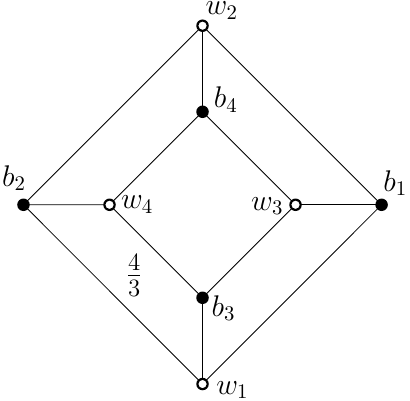}\hskip1cm\includegraphics[width=1.75in]{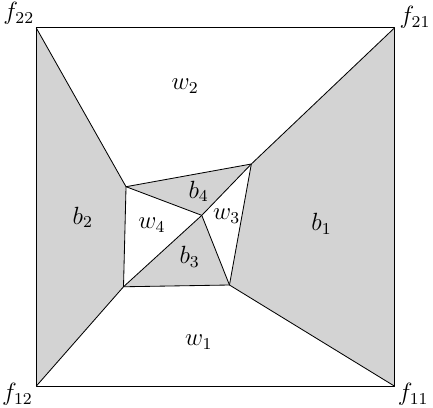} \hskip1cm \includegraphics[width=1.6in]{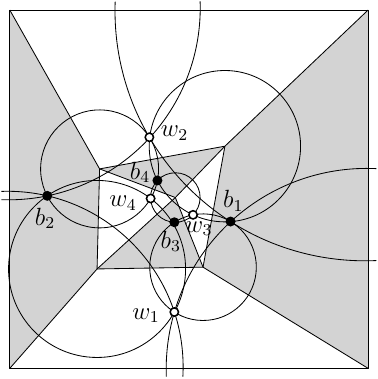}}
\caption{\label{quadexample}  The ``cube" graph is on the left with all face weights $1$ except the one indicated with weight $4/3$. Its augmented dual graph $\G^*$ is realized as circle centers having outer boundary a square $Q$ (middle). A circle pattern with these centers is shown on the right.}
\end{figure}

Let $\G$ have outer boundary vertices $w_1,b_1,w_2,b_2$. 
For  each bounded face $f$, let $X_f>0$ be its face weight.

The graph $\G^*$ has outer face of degree $4$; denote the vertices of its outer face by $f_{11}$, $f_{12}$, $f_{21}$ and $f_{22}$, where $f_{ij}$ is adjacent to the edge $(w_ib_j)^*$.

We construct a convex embedding in $Q$ of $\G^*$, with the outer vertices of $\G^*$ going to the vertices of $Q$, 
satisfying the property that the vertices of $\G^*$ go to the circle centers of a circle pattern with
the combinatorics of $\G$ (in the sense that the angles satisfy Proposition \ref{angles}), and moreover the face variables $X_f$
of $\G$ 
give the ``alternating product of edge lengths" as in (\ref{Xdef}). 

Let $W_1,B_1,W_2,B_2\in\C$ be the edges of $Q$ (summing to zero), oriented counterclockwise.
Let $K$ be a Kasteleyn matrix associated to $\G$ with face weights $X$.
Let $G(w)$ and $F(b)$ be functions on white and black vertices of $\G$
satisfying the properties that: for all internal white vertices $w$, we have
\be\label{cocycle1}\sum_b G(w) K_{wb}F(b) =0,\ee 
and for all internal black vertices $b$, we have
\be\label{cocycle2}\sum_w G(w) K_{wb}F(b) =0,\ee
and for $i=1,2$
\begin{align}
\label{bdycrit1}\sum_w G(w) K_{wb_i}F(b_i) &= B_i\\
\label{bdycrit2}\sum_b G(w_i) K_{w_ib}F(b) &= -W_i.
\end{align}
Functions $G,F$ satisfying (\ref{cocycle1}) and (\ref{cocycle2}) are said to give a \emph{Coulomb gauge} for $\G$. The reason for such a name is that the edge weights $G(w) K_{wb}F(b)$ have zero divergence at each internal vertex, which is similar to the case of the Coulomb gauge in electromagnetism, corresponding to the choice of a divergence-free vector potential \cite[Section 6.5]{Jackson1999}.
The existence of a Coulomb gauge $G,F$
taking the boundary values $(\vec B,-\vec W)$ (\ref{bdycrit1}),(\ref{bdycrit2})  (for graphs with boundary lengths $4$ or more) is discussed in Section~\ref{Coulombexistence} below.
As shown there, equations (\ref{cocycle1})-(\ref{bdycrit2}) 
determine $G$ and $F$ up to a finite number of choices: in fact one or two (and typically two) choices for boundary of length $4$, see below. 

Given $G,F$ satisfying the above, 
define a function $\omega$ on oriented edges by $\omega(w,b) = G(w)K_{wb}F(b)$
(and $\omega(b,w)=-\omega(w,b)$, so that $\omega$ is a flow, or 1-form).

Equations (\ref{cocycle1}) and (\ref{cocycle2}) imply that $\omega$ is co-closed (divergence free) at internal vertices.
Thus $\omega$ can be integrated to define a mapping $\phi$ from the augmented dual graph $\G^*$ into $\C$
by the formula 
\be\label{dphi}
\phi(f_1)-\phi(f_2) = \omega(w,b)\ee
where $f_1,f_2$ are the faces adjacent to edge $wb$, with $f_1$ to the left and
$f_2$ to the right when traversing the edge from $w$ to $b$. The mapping $\phi$ is defined up to an additive constant; by (\ref{bdycrit1}) and (\ref{bdycrit2}) 
we can choose the constant so that the vertices $f_{11},f_{12},f_{22},f_{21}$ go to the vertices of $Q$.

\begin{theorem}
\label{thm:planardimerstocircles}
Suppose $\G$ has outer face of degree $4$.
The mapping $\phi$ defines a convex embedding into $Q$ of $\G^*$ sending the outer vertices to the corresponding vertices of~$Q$. The images of the vertices of $\G^*$ are the centers of a circle pattern with the given combinatorics of $\G$ and face weights $X$.
Moreover the outer face of $\G$ will also be cyclic.
\end{theorem}

Boundary length $4$ is special in the sense that if $\G$ has outer face of degree strictly larger than $4$, 
the outer face of the associated circle pattern will not necessarily
be cyclic.

\begin{proof}

Our proof relies on two lemmas, whose statements and proofs are postponed until 
after the proof of the theorem. Lemma \ref{lem:postnikov} shows that $\G$ can be obtained from the $4$-cycle graph using a sequence of \emph{elementary transformations} (see Figure \ref{elemtransfs}). Lemma~\ref{lem:solutionssquare} shows that the theorem holds true when $\G$ is equal to the $4$-cycle graph. Therefore to complete the proof of the theorem it suffices to show that, if it holds for a graph then it holds for
any elementary transformation applied to that graph.

To use this argument we must extend slightly our notion of convex embedding to include the case 
when $\G$ has degree $2$ vertices, and when $\G$ has parallel edges, 
because these necessarily occur at intermediate stages when we build up the graph $\G$ from the $4$-cycle.

When $\G$ has parallel edges connecting two vertices $w$ and $b$, the graph $\G^*$ has one or more degree-$2$ vertices there; 
we assign a location to these vertices as shown in  Figure~\ref{parallel}.
Note that under the elementary move merging those edges of $\G$,
we can simply forget the corresponding circle and circle center of $\G^*$. 

\begin{figure}[htbp]
\center{\includegraphics[width=4.7in]{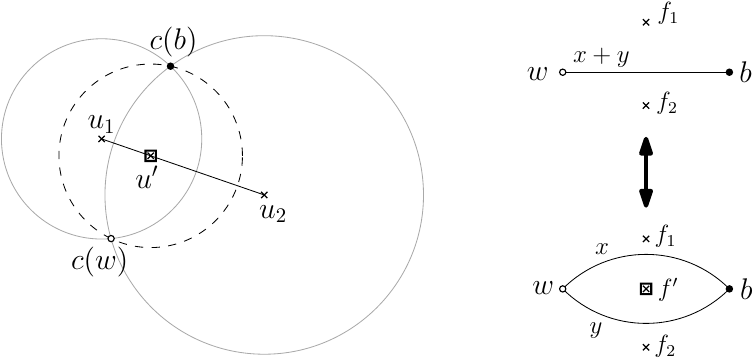}}
\caption{\label{parallel}
 Replacing a single edge with weight $(x+y)$ by parallel edges
with weights $x,y$ corresponds to adding a point~$u'=\phi(f')$ on segment~$u_1u_2$, where~$u_{1}=\phi(f_{1})$ and $u_{2}=\phi(f_2)$. The position of the point~$u'$ on the segment is chosen so that~$\frac{|u_1-u'|}{|u_2-u'|}=\frac xy$. The circle (in dashed stroke) corresponding to~$f'$ is centered at~$u'$ and passes through~$c(w)$ and~$c(b)$.}

\end{figure}

When $\G$ has a degree-$2$ vertex $v$, connected to neighbors $v_1$ and $v_2$, 
then for the associated Coulomb $1$-form $\omega$ we necessarily have $\omega{(v,v_1)}+\omega{(v,v_2)}=0$.
This implies that under $\phi$ the duals of these edges get mapped to the same edge. We call this a ``near-embedding"
since faces of degree $2$ in $\G^*$ get collapsed to line segments. 
Note however that for any such graph $\G$,
contracting degree-$2$ vertices results in a new graph with the same mapping $\phi$, minus those paired edges.

Consequently, among the elementary transformations of Figure \ref{elemtransfs}, only the spider move
has a nontrivial effect on the embedding.

Now let $\H$ be a graph obtained from $\G$ by applying a 
spider move. The embedding of $\H^*$ is obtained from the embedding of $\G^*$ by a ``central move", see equation (\ref{eq:quadeq}) and Figure~\ref{fig:Miquelcenters} in Section \ref{sec:Miquel}.
This move gives a convex embedding by convexity of the faces: the new central vertex is necessarily in the convex
hull of its neighbors: see Lemma \ref{lem:solutionssquare} below.

For the case $\G$ is a $4$-cycle, see Corollary~\ref{cor} and Figure~\ref{square}.
\begin{figure}[htbp]
\center{\includegraphics[width=3.5in]{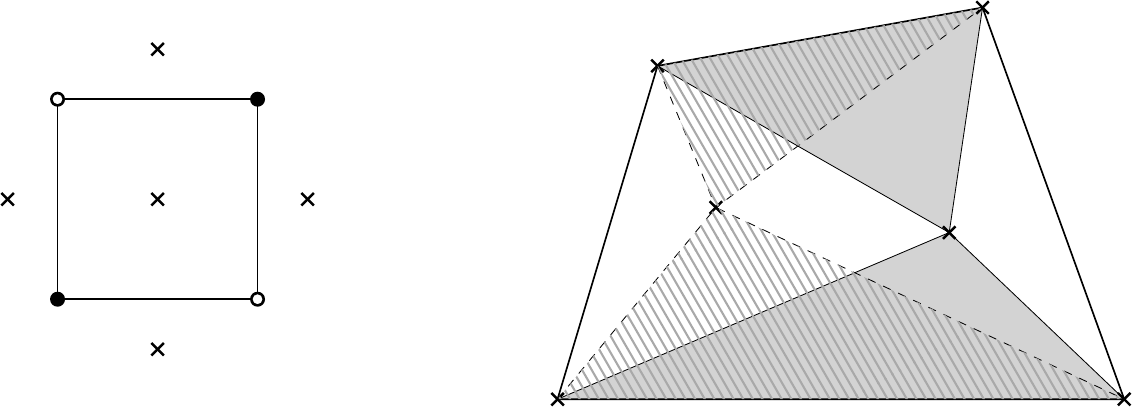}}
\caption{\label{square} For the $4$-cycle
there are generically two Coulomb gauges and thus two embeddings of $\G^*$, see Corollary \ref{cor}. These are shown for a particular choice of boundary.}
\end{figure}

Finally, the fact that $\phi$ maps the vertices of $\G^*$ to centers of a circle pattern follows from the proof of 
Proposition \ref{angles}
and the fact that the sum of the angles of the corners of the white/black faces around a given vertex of $\phi(\G^*)$ equals $\pi$. 
The fact that the outer face is cyclic also follows by induction: this is true for the $4$-cycle, and the central moves
do not move the outer dual vertices, or change their radii.
\end{proof}

We now show that $\G$ can be reduced to the $4$-cycle graph using elementary transformations.

\begin{lemma}
\label{lem:postnikov}
Let $\G$ be a finite connected nondegenerate planar bipartite graph with $4$ marked boundary vertices (two black and two white, with colors alternating while going around the outer face). Then $\G$ can be reduced to the $4$-cycle graph by applying a sequence of elementary transformations described in Figure \ref{elemtransfs}, without modifying the $4$ marked vertices at intermediate stages.
\end{lemma}

\begin{proof}
We rely on the theory of planar bicolored graphs in the disk (also known as plabic graphs) developed by Postnikov \cite{postnikov2006}, of which finite planar bipartite graphs with marked boundary vertices are a special case: just draw the graph inside a disk and attach each boundary vertex to the boundary of the disk using an additional edge. We declare two graphs to be \emph{equivalent} if one can get from one to the other using a sequence of spider moves and degree $2$ vertex contractions or their inverses. A graph is said to be \emph{reduced} if there is no graph in its equivalence class to which one can apply the merging of parallel edges. By applying merging of parallel edges as much as necessary, we first transform $\G$ via a sequence elementary transformations into a reduced graph $\G'$. Note that $\G'$ is connected and nondegenerate, since these properties are preserved by elementary transformations.

We define a \emph{zigzag path} of $\G'$ to be any path starting at the boundary of the disk and turning maximally left (resp. maximally right) at each white (resp. black) vertex. By \cite[Theorem 13.2]{postnikov2006}, since $\G'$ is reduced, every zigzag path ends on the boundary of the disk. Label the boundary points of the disk cyclically and define the \emph{boundary permutation} $\pi$ by setting $\pi(i)=j$ if the zigzag path starting at the $i$th boundary point ends at the $j$th boundary point. By \cite[Theorem 13.2]{postnikov2006}, since $\G'$ is connected, $\pi$ has no fixed point. It follows from \cite[Theorem 13.4]{postnikov2006} that $\G'$ is equivalent to another reduced graph $\G''$ if and only if they have the same boundary permutation. The boundary permutation of the plabic graph associated with the four-cycle graph is $(13)(24)$, where we write permutations as products of cycles with disjoint supports. We will show that all the other permutations of four elements cannot arise as the boundary permutation of $\G'$. Figure \ref{fig:plabic} displays three reduced graphs associated respectively with the three permutations $(1234)$, $(1324)$ and $(12)(34)$. These graphs cannot be equivalent to $\G'$ since they are respectively degenerate for the first two and not connected for the third. All the remaining permutations are obtained by symmetry from these three permutations, hence the boundary permutation of $\G'$ has to be $(13)(24)$. 
\end{proof}

\begin{figure}[htbp]
\center{\includegraphics[width=1.5in]{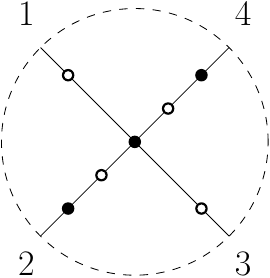}\hskip1cm\includegraphics[width=1.6in]{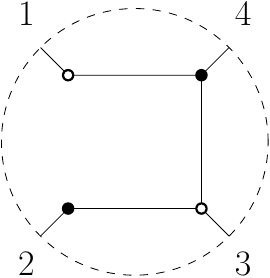}}\hskip1cm\includegraphics[width=1.5in]{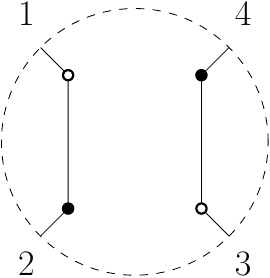}
\caption{\label{fig:plabic}Three reduced graphs with respective boundary permutations $(1234)$, $(1324)$ and $(12)(34)$, from left to right. The graph on the left has no dimer cover since it has more white vertices than black vertices, the graph in the middle has an edge which is used by no dimer cover and the graph on the right is not connected.}
\end{figure}

We now treat the base case of the induction in the above proof, namely the case of the $4$-cycle. It has one inner face with face weight $X$. By an Euclidean motion and scaling, we can assume the four outer vertices of $\G^{*}$ are placed at locations $0,1,\rho, \xi \in \mathbb{C}$. It remains to determine the only inner vertex $u$ of  $\G^{*}$.

\begin{lemma}
\label{lem:solutionssquare}
Let $Q'$ be a convex quadrilateral with vertices $0,1,\rho,\xi$ in counterclockwise order
and let $X\in(0,\infty)$. The equation 
\be\label{ueq}-\frac{(1-u)(\xi-u)}{(0-u)(\rho-u)}=X\ee
has two solutions $u$ (counted with multiplicity), both of which lie strictly inside $Q'$.
\end{lemma}

\begin{proof}[Proof of Lemma~\ref{lem:solutionssquare}] When $X=0$ the solutions to (\ref{ueq}) are $u=1,u=\xi$ and when $X=\infty$
the two solutions are $u=0,u=\rho$. Notice that no other point on the boundary of $Q$ can be a solution for any $X$
because 
one of the angle sums $0u1+\rho u\xi$ or $1u \rho+\xi u0$ would be larger than $\pi$. So by continuity it suffices
to show that for small $X>0$ there is one solution inside $Q'$ near $1$ and one solution inside $Q'$ near $w$.
Solving (\ref{ueq}) for $u$ and expanding near $X=0$ gives the solutions
\begin{align*}u&=1-\frac{\rho-1}{\xi-1}X + O(X^2)\\
u&=\xi-\xi\frac{\rho-\xi}{1-\xi}X + O(X^2).
\end{align*}
Note that for the first solution, $\arg\frac{\xi-1}{\rho-1}$ is less than the angle at $1$ of $Q'$, so the vector
$-\frac{\rho-1}{\xi-1}$ points into $Q'$ from the point $1$; thus this solution
is inside $Q'$ for small $X>0$. For the second, $\arg\frac{\rho-\xi}{1-\xi}$ is less than the angle of $Q'$ at 
$\xi$,
so the vector $-\xi\frac{\rho-\xi}{1-\xi}$ points into the interior of $Q'$ from $\xi$. 
\end{proof}

Each solution $u$ determines a Coulomb gauge for the 4-cycle graph sending the outer vertices to $Q'$, and hence to any convex quadrilateral $Q$:

\begin{corollary}\label{cor}
	Suppose that $\G=\{w_1, b_1, w_2, b_2\}$ is a single 4-cycle with the inner face~$f$. Let  $u$ solve (\ref{ueq}) and suppose the convex quadrilateral $Q$ differs from $Q'$ by a complex dilation $\lambda\in \mathbb{C}$. Then the pair of functions~$G: \{w_1,w_2\}\to\mathbb{C}$ and~$F:\{b_1,b_2\}\to\mathbb{C}$ defined by
	\begin{align*}
	&G(w_1):=\lambda,\\
	&F(b_1):=\lambda\frac{0-u}{K_{w_1b_1}}, \quad
	F(b_2):=\lambda\frac{\xi-u}{K_{w_1b_2}},\\
	&G(w_2):=\lambda\frac{\rho-u}{F(b_2)K_{w_2b_2}}=\lambda \frac{1-u}{F(b_1)K_{w_2b_1}},
	\end{align*}
	gives a Coulomb gauge for $\G$; and the mapping~$\phi$ defines a convex embedding of~$\G^*$ into~$Q$, where $\phi(f)=u$.
\end{corollary}

\begin{remark} 
The \emph{isogonal conjugate} of a point $U$ with respect to a quadrilateral  $ABCD$ is constructed 
by reflecting the lines $UA$, $UB$, $UC$  and $UD$ about the angle bisectors of $A$, $B$, $C$, and $D$ respectively. 
If these four reflected lines intersect at one point, then this point is called the isogonal conjugate of $U$. 
Not all points have an isogonal conjugate with respect to a quadrilateral, but only those lying on a certain
cubic curve associated with the quadrilateral~\cite{Akopyan2007}.
One can show that the two solutions to (\ref{ueq}) are isogonally conjugate with respect to $Q'$. 
Moreover all possible pairs of isogonal conjugate points inside $Q'$ can be achieved upon varying $X>0$.
\end{remark}

\subsection{Existence of Coulomb gauge}
\label{Coulombexistence}

In this subsection we consider the case when the outer face of the planar bipartite graph has an arbitrary degree. 
We prove the existence of at least one Coulomb gauge in this setting.

\begin{lemma} 
Let $\G$ have outer boundary vertices $w_1,b_1, \ldots, w_k,b_k$. Let $Q\subset\C$ be a convex polygon, with non-zero edges $W_1,B_1, \ldots, W_k,B_k \in\C$ summing to zero. 
Suppose that $(\G\smallsetminus\{w_1,b_1, \ldots, w_k,b_k\})$ has a dimer covering and $(\G\smallsetminus\{w_1,b_1, \ldots, w_k,b_k\})\cup\{b_i,w_j\}$ for any $b_i, w_j$ on the boundary has a dimer covering as well. Then there exists a solution to (\ref{cocycle1})-(\ref{bdycrit2}) for boundary values $(\vec B,-\vec W)$.\end{lemma}

\begin{proof} 
It suffices to assume that $Q$ is \emph{generic}, that is its sides have no algebraic relations
beyond summing to zero. Indeed, a limit of dual embeddings of $\G^*$ into nearby generic 
polygons $Q$ will be
a dual embedding into $Q$.

Let $K$ be a Kasteleyn matrix of $\G$.
Note that $K$ is invertible. 
Let $G(w)$ and $F(b)$ be functions on white and black vertices of $\G$ defined by
\[
G(w)=\sum_{i = 1}^k\alpha_iK^{-1}_{b_iw}  \quad \text{and} \quad 
F(b)=\sum_{j = 1}^k K^{-1}_{bw_j}\beta_j
\]
for some $\alpha_i,\beta_j \in \mathbb{C}$. For all internal white vertices $w$, we have
\[\sum_b G(w) K_{wb}F(b) =0,\]
and for all internal black vertices $b$, we have
\[\sum_w G(w) K_{wb}F(b) =0,\]
and for all $i,j\in\{1,\ldots, k\}$
\begin{align*}
\sum_w G(w) K_{wb_i}F(b_i) &= \sum_{j = 1}^k \alpha_iK^{-1}_{b_iw_j}\beta_j\\
\sum_b G(w_j) K_{w_jb}F(b) &= \sum_{i = 1}^k \alpha_iK^{-1}_{b_iw_j}\beta_j.
\end{align*}

Let $M$ be the $k\times k$ matrix defined by $m_{ij}=K^{-1}_{b_iw_j}$.  By the generalized Cramer's rule
\[
\det M= \pm \frac{\det(K_{\G\smallsetminus\{w_1,b_1, \ldots, w_k,b_k\}})}{\det K},
\]
where $K_{\G\smallsetminus\{w_1,b_1, \ldots, w_k,b_k\}}$ denotes the submatrix of $K$ formed by choosing the rows indexing the white vertices of $\G\smallsetminus\{w_1,b_1, \ldots, w_k,b_k\}$ and columns indexing the black vertices of $\G\smallsetminus\{w_1,b_1, \ldots, w_k,b_k\}$. Note that  $K_{\G\smallsetminus\{w_1,b_1, \ldots, w_k,b_k\}}$ is a Kasteleyn matrix of the graph $\G\smallsetminus\{w_1,b_1, \ldots, w_k,b_k\}$, hence
\[
|\det(K_{\G\smallsetminus\{w_1,b_1, \ldots, w_k,b_k\}})|= Z_{\G\smallsetminus\{w_1,b_1, \ldots, w_k,b_k\}} = \sum_{m\in M_{\G\smallsetminus\{w_1,b_1, \ldots, w_k,b_k\}}}\nu(m).
\]
Therefore the condition that $(\G\smallsetminus\{w_1,b_1, \ldots, w_k,b_k\})$ has a dimer covering implies that  $\det M\neq 0$, i.e. the matrix $M$ is invertible.
Similarly, since $(\G\smallsetminus\{w_1,b_1, \ldots, w_k,b_k\})\cup\{b_i,w_j\}$ has a dimer covering for any $b_i, w_j$ on the boundary, all elements of the inverse matrix $M^{-1}$ are non-zero ($m^{-1}_{ji}\neq0$). We are looking for $k$-tuples $\vec \alpha=(\alpha_1,\dots,\alpha_k)$ and $\vec \beta=(\beta_1,\dots,\beta_k)$ such that 
\[
 \sum_{j = 1}^k \alpha_im_{ij}\beta_j=B_i \quad \text{and} \quad 
\sum_{i = 1}^k \alpha_im_{ij}\beta_j=-W_j.
\]

The existence of a solution to a general problem of this type for generic $\vec B,\vec W$ is given in the appendix.
\end{proof}

\noindent{\bf Question:} If there are $2k$ outer boundary vertices (where $k>2$), how many solutions are there with embedded centers? For generic weights this number is a function only of the graph $\G$, and is
invariant under elementary transformations preserving the boundary.

\section{Biperiodic bipartite graphs and circle patterns}\label{sectionbiperiodic}

In this section we deal with the case of a bipartite graph embedded on a torus, or equivalently a biperiodic bipartite 
planar graph.

\subsection{Embedding of $\G^*$}
Let $\G$ be a bipartite graph, having dimer covers, nondegenerate and embedded on a torus $T$ with complementary regions (faces) which are topological disks.
Let $\nu:E\to\R_{>0}$ be a set of positive edge weights. 
We fix two simple cycles $l_1$ and $l_2$ in the dual graph $\G^*$ which together generate the homology $H_1(T,\Z)$, and have intersection number 
$l_1\wedge l_2=+1$. Let $K$ be a real Kasteleyn matrix.

We define new edge weights on $\G$
by multiplying, for $i=1,2$, each original edge weight by $\lambda_i$ (resp. $\lambda_i^{-1}$) if the edge 
crosses $l_i$ with a white
vertex to its left (respectively right).
Define $K(\lambda_1,\lambda_2)$ to be a Kasteleyn matrix of $\G$ with the new edge weights 
and define a Laurent
polynomial $P$ by $P(\lambda_1,\lambda_2):=\det K(\lambda_1,\lambda_2)$. Note that the graph $\G$ with these new edge weights still has the \emph{same face weights} as the original graph.
We shall see exactly for which $(\lambda_1,\lambda_2)\in(\C^*)^2$ there is a corresponding circle pattern (Theorem \ref{thm:correspondence} below).

The \emph{spectral curve 
of the dimer model on $\G$} is defined to be the zero-locus of~$P$ in $(\C^*)^2$. 
The \emph{amoeba} of $P$ is the image in $\R^2$ of the spectral curve 
under the mapping $(\lambda_1,\lambda_2)\mapsto(\log|\lambda_1|,\log|\lambda_2|).$ 
The spectral curve is a \emph{simple Harnack curve}~\cite{KOS}; this has the following consequences (see \cite{Mikh}). 
Every point $(\lambda_1,\lambda_2)$ of the spectral curve is a simple zero of $P(\lambda_1, \lambda_2)$ 
or it is a double zero which is real (a real node). 
The partial derivatives $P_{\lambda_1}$ and  $P_{\lambda_2}$ vanish only at real points, and vanish
simultaneously only at real nodes;
the quantity $\zeta:=\frac{\lambda_2P_{\lambda_2}}{\lambda_1P_{\lambda_1}}$ is the \emph{logarithmic slope} and is real
exactly on the boundary of the amoeba (where it equals the slope of the amoeba boundary). 
At a real node the logarithmic slope $\zeta$ has exactly two 
limits, which are nonreal and conjugate.

When  $(\lambda_1,\lambda_2)$ is a simple zero, $(\bar\lambda_1,\bar\lambda_2)$ is also a zero of $P(\lambda_1,\lambda_2)$, and in this case it is shown in \cite{KO} that
$K(\lambda_1,\lambda_2)$ has a kernel which is one-dimensional. 
Hence there exists a pair of functions $(F,G)$ unique up to scaling, with $F$ defined on black vertices
and $G$ defined on white vertices,  with $F\in \ker K(\lambda_1, \lambda_2)$ and $G\in \ker K^t(\lambda_1, \lambda_2)$.
When $\lambda_1,\lambda_2$ are not both real we call it an \emph{interior} simple zero,
it corresponds to a point in the interior of the amoeba, but not at a node. 

At a real node, $\lambda_1$ and $\lambda_2$ are both real and the kernel of $K(\lambda_1, \lambda_2)$ is two-dimensional. The kernel is spanned by the limits of the kernels for nearby simple zeros and their conjugates.
Let $F,G$ be functions in the kernel of $K(\lambda_1, \lambda_2)$ (resp. of $K^t(\lambda_1, \lambda_2)$)
which are limits of those for simple zeros for which $\Im\zeta>0$.

Let $\tilde{\G}$ be the lift of $\G$ to the plane (the universal cover of the torus). 
Let $p_1, p_2$ be the horizontal and vertical periods of $\tilde{\G}$ corresponding to $l_1,l_2$ respectively. 
We lift $F$ and $G$ to $\tilde{\G}$ by, for $i=1,2$
\be\label{FGplane}
F(b+p_i)=\lambda_iF(b), \quad G(w+p_i)=\lambda^{-1}_iG(w).
\ee
The extended version of $F$ (resp. $G$) is in the kernel of $K_{wb}$ (resp. $K_{wb}^t$), where 
 \[K_{wb}=\begin{cases}
K_{wb}(\lambda_1,\lambda_2) & \text{if the edge } wb \text{ belongs to a fundamental domain} \\    
  \lambda_i^{-1}\cdot K_{wb}(\lambda_1,\lambda_2) & \text{if the edge } wb \text{ crosses } l_i  \text{ with a white vertex to its left}\\
   \lambda_i\cdot  K_{wb}(\lambda_1,\lambda_2) & \text{if the edge } wb \text{ crosses } l_i  \text{ with a white vertex to its right}\\
     0  & \text{if } w \text{ and } b \text{ are not adjacent}\\
   K_{w'b'} & \text{if } w'=w+p_i \text{ and } b'=p+p_i 
\end{cases}\]
is a real-valued biperiodic Kasteleyn matrix of $\tilde{\G}$.

We define two co-closed 1-forms
\be\label{omegaomega}
\omega(w,b)=G(w)K_{wb}F(b)\quad \text{ and } \quad \hat\omega(w,b)=\overline{G(w)}K_{wb}F(b)
\ee
and use them to define two mappings $\phi,\hat\phi:\tilde\G^*\to \C$ using (\ref{dphi}).

\begin{remark}
\label{rem:phiperiodic}
The image of the mapping $\phi$ is periodic, in the sense that (for $i=1,2$)
$\phi(v+p_i) = \phi(v)+ V_i$
for constant vectors $V_1,V_2$ called the periods. Indeed, 
\be\label{translate}F(b+p_{1})G(w+p_{1})=F(b)G(w)\ee
and similarly for $p_2$. In the case of a real node, $\lambda_1$ and $\lambda_2$ are real hence \eqref{translate} also holds with $G$ replaced with $\overline G$ so $\hat \phi$ also has a periodic image. 
\end{remark}

As a consequence of Remark~\ref{rem:phiperiodic}, one can project the centers $\{\phi(\tilde\G^*)\}$
down to a flat torus and one can give an explicit formula for the aspect ratio of that torus:

\begin{proposition}\label{logslope} The periods $V_1$ and $V_2$ of $\phi$ are nonzero as long as $P_{\lambda_1}$ and  $P_{\lambda_2}$ are nonzero, and can also be chosen nonzero at a real node by an appropriate scaling. 
The ratio of the periods is $V_2/V_1 = \zeta = \frac{\lambda_2P_{\lambda_2}}{\lambda_1P_{\lambda_1}}.$
\end{proposition}
\begin{proof}
For a matrix $M$ we have the identity 
$\frac{\partial(\det M)}{\partial M_{i,j}} =M^*_{i,j}$, where $M_{i,j}$ is the $(i,j)$-entry of $M$ and $M^*_{i,j}$
is the corresponding cofactor. 
Recalling that $P=\det K$, we have
$$\lambda_1 \frac{\partial P}{\partial \lambda_1} = \sum_{\gamma_1} \pm K_{wb}K^*_{bw}$$
where the sum is over edges $wb$ crossing $\gamma_1$ (i.e. those edges of $\G$ with a weight 
involving $\lambda_1^{\pm 1}$),
the sign is given by the corresponding exponent of $\lambda_1$ for that edge, and $K^*$ is the cofactor matrix. 
When $(\lambda_1,\lambda_2)$ is a simple zero of $P$, we have $K^* K=K K^* = (\det{K})\, \mbox{Id}  = 0$ and hence the columns of $K^*$ are multiples of $F$ and the rows are multiples of $G$. In particular, we can write $K^*_{bw} = \operatorname{cst}\cdot F(b)G(w)$ for some scale factor~$\operatorname{cst}$.
We find
$$\lambda_1 \frac{\partial P}{\partial \lambda_1} = \operatorname{cst}\sum_{\gamma_1} \pm K_{wb}F(b)G(w) = \operatorname{cst}\cdot V_1.$$
Similarly
$$\lambda_2 \frac{\partial P}{\partial \lambda_2} = \operatorname{cst}\cdot V_2$$
and we conclude by taking the ratio of these.

If $(\lambda_1,\lambda_2)$ is a node, we can take a limit of nearby simple zeros with, say, $\Im(\zeta)>0$,
and scaling so that $V_1$ is of constant length;
since $\zeta$ has a well-defined nonreal limit, $V_2$ will also have a limit of finite length.  
\end{proof}

\begin{theorem}
\label{thm:periodicdimers} 
If $(\lambda_1,\lambda_2)$ is in the interior of the amoeba of $P$,
the realization $\phi$ is a periodic convex embedding of $\tilde \G^*$,
dual to a circle pattern.
\end{theorem}

Note that if $(\lambda_1,\lambda_2)$ is on the \emph{boundary} of the amoeba, then $\zeta$ is real and so by Proposition 
\ref{logslope}, $\phi$ cannot define an embedding.

\begin{proof}
If $(\lambda_1,\lambda_2)$ is an interior simple zero, then we show in Lemma \ref{tgraph} below that the realization $\phi_1$ 
defined from $\Re(G(w))K_{wb}F(b)$ is a ``T-graph embedding" (see definition in Section~\ref{subsec:periodicTgraphs}), mapping each white face to a convex polygon. (This result is stated in \cite{Kenyon2004} without proof). 
In particular for $\phi_1$
the sum of the angles of white polygons
at vertices of $\tilde\G^*$
is $\pi$. This implies that for the realization defined by $\phi$, the sum of angles of the white polygons
at vertices of $\tilde\G^*$ is also $\pi$, since these polygons are simply scaled copies of those for $\phi_1$. 
Likewise the realization $\phi_2$ 
defined from $G(w)K_{wb}\Re(F(b))$ is a T-graph embedding, mapping each black face to a convex polygon. 
It suffices to show that the orientations of $\phi_1$ and $\phi_2$ agree. Note that the realization defined
by $\overline{G(w)}K_{wb}\Re(F(b))$ is also a T-graph embedding with the \emph{reverse} orientation to that of $\phi_2$.
Thus the orientations of the white and black faces agree in exactly one of $\phi$ or $\hat\phi$.
We claim that they agree in $\phi$, not $\hat\phi$. This is a consequence of Lemma~\ref{bound_simple} 
below.
Thus $\phi$ is a local homeomorphism.
 
By Remark~\ref{rem:phiperiodic} and Proposition \ref{logslope}, the image of $\phi$ is periodic with nonzero periods (at interior simple zeros 
the $P_{\lambda_i}$ are nonzero) and so 
$\phi$ is proper, and thus a global embedding (a proper
local homeomorphism is a covering map). 

If $(\lambda_1,\lambda_2)$ is a real node then one can argue similarly as above;
the question of orientation is resolved by taking a limit of simple zeros, since the embeddings depend continuously on 
$(\lambda_1,\lambda_2)$. 
\end{proof}

\subsection{The circles}

Let $\phi$ be the embedding of $\G^*$ defined from $\omega$ in (\ref{omegaomega}), and $\hat\phi$ the realization defined 
from $\hat\omega$. From $\phi$ we can define the associated circle pattern as in the finite case. However the circles may  grow without bound in radius as we move away from the initial vertex; we give
a criterion (Lemmas \ref{phihatbdd}, \ref{bound_simple} and \ref{lem:doublepoint}) for 
determining when the radii are bounded. 

\begin{lemma}
\label{phihatbdd}
The boundedness of the map $\hat\phi$ is equivalent to the boundedness of the radii in any circle pattern.
\end{lemma}

\begin{figure}[htbp]
\center{\includegraphics[width=4.7in]{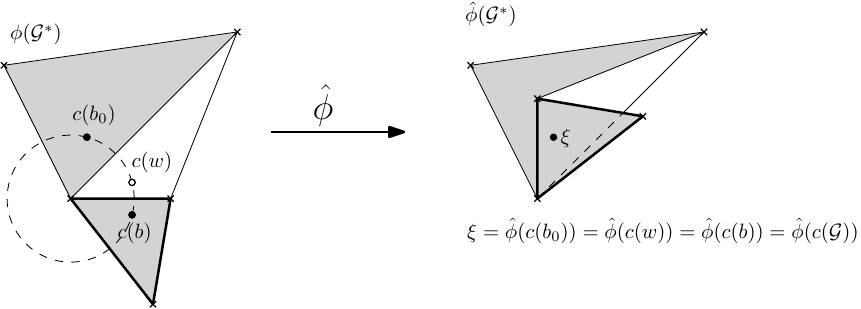}}
\caption{\label{folding_map} Under folding, all vertices of a circle pattern coincide.}
\end{figure}

\begin{proof}
First, note that $\hat\phi(\G^*)$ is defined up to an additive constant. For any face $v$ of $\G^*$ one can describe a face $\hat\phi(v)$ using the following procedure: fix a face $\phi(b_0)$, consider a face-path $g$ from $\phi(v)$ to $\phi(b_0)$ and take a sequence of reflections of $\phi(v)$ across edges of $\phi(\G^*)$ intersecting $g$. Note that $\hat\phi(v)$ is independent of the choice of a face-path, since $\phi(\G^*)$ satisfies the angle condition around each vertex. In other words to get $\hat\phi(\G^*)$ from $\phi(\G^*)$ one can chose a root face $\phi(b_0)$ and fold the plane along every edge of the embedding, see Figure~\ref{folding_map}. Note that all black faces of $\hat\phi(\G^*)$ have their orientation preserved while white ones are reversed. 
Note also that two adjacent vertices of a circle pattern corresponding to $b, w\in\G$ are symmetric with respect to the edge $\phi((wb)^*)$. Therefore they coincide after one folds the plane along $\phi((wb)^*)$. Hence each circle pattern corresponds to a single point under the mapping $\hat\phi$, and the radii in the circle pattern are distances from this point to vertices of  $\hat\phi(\G^*)$. To finish the proof note that the boundedness of these distances is equivalent to the boundedness of the map~$\hat\phi$.
\end{proof}

We now explain when $\hat\phi$ is bounded.

\begin{lemma}\label{bound_simple}
If $(\lambda_1,\lambda_2)$ is an interior simple zero, then $\hat\phi$ is bounded.
\end{lemma}

\begin{proof}
Assume first that neither of $\lambda_1,\lambda_2$ is real. Fix a dual vertex $f\in\tilde\G^*$. We have \[
\hat{\phi}(f+2 p_1)-\hat{\phi}(f+p_1)=\lambda_{1}\bar\lambda^{-1}_{1} (\hat{\phi}(f+ p_1)-\hat{\phi}(f)).\]  
Since $|\lambda_1\bar\lambda^{-1}_1|=1$, the segment $\hat{\phi}(f+p_1)\hat{\phi}(f+2p_1)$ differs 
from $\hat{\phi}(f)\hat{\phi}(f+p_1)$ by a rotation around the center of the circle $C_{\phi(f),1}$ 
through $\hat{\phi}(f),\hat{\phi}(f+p_1),\hat{\phi}(f+2p_1)$ with angle $\arg \lambda_{1}\bar\lambda^{-1}_{1} \neq 0$. In particular this implies all the $\hat\phi(f+kp_1)$ for $k\in\Z$ lie on $C_{\phi(f),1}$. 
A similar argument holds for $\hat\phi(f+kp_2)$. 
So all the dual vertices $\hat\phi(f+k_1p_1+k_2p_2)$ with $(k_1,k_2)\in\Z^2$ have distance at most the sum of the diameters of these two circles from $\hat\phi(f)$ and thus lie in a compact set.

Assume now that $\lambda_1$ is real and $\lambda_2$ is non-real; then the image of $\hat\phi$ is periodic in the direction of $p_1$ and almost periodic in the direction of $p_2$. 
We claim that this is possible only if the period in the direction of $p_1$ is zero: 
on the one hand the four points $\hat\phi(f)$, $\hat\phi(f+p_1)$, $\hat\phi(f+p_2)$ and $\hat\phi(f+p_1+p_2)$ form a parallelogram (maybe degenerate) because of the periodicity in the direction of $p_1$, and on the other hand, the vectors $\hat\phi(f)\hat\phi(f+p_1)$ and $\hat\phi(f+p_2)\hat\phi(f+p_1+p_2)$ differ by multiplication by $\lambda_2\bar\lambda_2^{-1}\neq1$, so these vectors must be zero. Therefore $\hat\phi$ is also bounded in this case.
\end{proof}

\newcommand{\uu}{{\tau}}
\newcommand{\vv}{{\eta}}
\newcommand{\aaa}{{V_{\phi}}}
\newcommand{\aA}[1]{V_{\raisebox{-1pt}{\tiny#1}}}
\newcommand{\bB}[1]{V^{\star}_{\raisebox{-1pt}{\tiny#1}}}
\newcommand{\bbb}{{V_{\widehat\phi}}}

\begin{lemma}
\label{lem:doublepoint}
For real nodes on the spectral curve, there is a two-parameter family of embeddings $\phi$, up to similarity, 
but exactly one of them has a bounded $\hat\phi$. Moreover, boundedness of 
$\hat\phi$ (in this case) is equivalent to the biperiodicity of the radii in any circle pattern.
\end{lemma}

\begin{proof} This proof is due to Dmitry Chelkak. By Remark~\ref{rem:phiperiodic}, the images of $\phi$ and $\hat\phi$ are periodic in the case of a real node.
Denote by $\aA{1},\aA{2}$ the corresponding periods  of $\phi$ and by $\bB{1},\bB{2}$ 
the periods of the map $\hat\phi$. Note that for each $\uu, \vv$ in the unit disk $\D=\{z : |z|<1\}$ the pair of functions $(F+\uu\bar F, G+\vv\bar G)$ 
also defines, via (\ref{dphi}), a non-degenerate embedding $\phi_{\uu,\vv}$:
a black face of $\phi_{\uu,\vv}$ is the image of a black face $b$ of $\phi$ under the linear map $z\mapsto z+\vv\bar z$,
followed by a homothety with factor $\frac{F(b)+\uu\bar F(b)}{F(b)}$.
Similarly the white faces undergo the linear map $z\mapsto z+\uu\bar z$ followed by a homothety; these 
linear maps have positive determinant when $\uu,\vv\in\D$, and so preserve orientation (and convexity).  
Let $\hat\phi_{\uu,\vv}$ be a map defined via (\ref{dphi}) by the 
pair of functions $(F+\uu\bar F, \overline{G+\vv\bar G}).$

If we let $\aaa=n\aA{1}+m\aA{2}$ be a period of $\phi$ and $\bbb=n\bB{1}+m\bB{2}$ be the corresponding period of $\hat\phi$, 
then the corresponding period of $\phi_{\uu,\vv}$ is $\aaa+\uu\overline \bbb+\vv \bbb +\uu \vv\overline \aaa$. In order for $\phi_{\uu,\vv}$ to be embeddings for all $\uu, \vv \in\D$ we need
\be\label{abcondition}
\aaa+\uu\overline\bbb+\vv \bbb +\uu\vv\overline \aaa \neq 0 \,\text{ for all }\, \uu, \vv \in\D.
\ee

Note that $\aaa+\uu\overline\bbb+\vv \bbb +\uu\vv\overline \aaa=0$ when $\uu=-\frac{\bbb\vv+\aaa}{\vv\overline \aaa+\overline\bbb}$. Under what conditions
are there no solutions with $\uu,\vv\in \D$? 
The map $\vv \mapsto -(\bbb\vv+\aaa) / (\bar{\aaa}\vv+\overline\bbb)$ sends the unit circle to itself, and maps the unit disk strictly
outside the unit disk if and only if $|\aaa|\ge |\bbb|$.
So the above condition (\ref{abcondition}) is equivalent to the condition $|\aaa|\geq|\bbb|$ for all~$n, m$. Note that a period of $\hat\phi$ does not exceed in absolute value a period of $\phi$, to see this recall the folding interpretation of $\hat\phi$ described in the proof of Lemma~\ref{phihatbdd}.
This condition can be further reformulated as follows: the image of $\mathbb{D}$ under the mapping $z\mapsto (\aA{1}+\bB{1}z)/(\aA{2}+\bB{2}z)$ does not intersect the real line: otherwise, one would have $(\aA{1}-t\aA{2})+(\bB{1}-t\bB{2})z=0$ for some real $t$,
a contradiction with $|\aaa|\geq|\bbb|$.

Recall that the image of $\hat\phi_{\uu,\vv}$ is periodic, hence it is bounded if and only if its periods are zero. Thus we need to find $\uu,\bar \vv\in\mathbb{D}$ such that:
\begin{align*}
\bB{1}+\uu\overline{\aA{1}}+\bar \vv \aA{1}+\uu\bar{\vv}\overline{\bB{1}}&=0\\
\bB{2}+\uu\overline{\aA{2}}+\bar \vv \aA{2}+\uu\bar{\vv}\overline{\bB{2}}&=0
\end{align*}
or equivalently,
\[
\frac{\bB{1}+\bar \vv \aA{1}}{\overline{\aA{1}}+\bar \vv \overline{\bB{1}}}=-\uu=\frac{\bB{2}+\bar \vv \aA{2}}{\overline{\aA{2}}+\bar \vv  \overline{\bB{2}}}.
\]
Both fractional-linear mappings send the unit disk to itself, since $|\aA{1,2}|\geq|\bB{1,2}|$. Therefore it is enough to show 
that this quadratic equation in $\bar \vv$ has a root in $\D$; the corresponding $\uu$ will lie in $\D$ too.

Clearly, either one of the roots is inside the unit disk and the other outside, or both are on the unit circle. Finally, note that the latter is impossible as one would have 
\[
\frac{\overline{\aA{1}}+\bar \vv\overline{\bB{1}}}{\overline{\aA{2}}+\bar \vv \bar b_2}=\frac{\bB{1}+\bar \vv \aA{1}}{\bB{2}+\bar \vv \aA{2}}=\frac{ \aA{1}+\vv \bB{1}}{\aA{2}+ \vv \bB{2}}\in\mathbb{R}
\]
which is in contradiction with the fact that the image of $\mathbb{D}$ under the mapping $z\mapsto (\aA{1}+\bB{1}z)/(\aA{2}+\bB{2}z)$ does not intersect the real line.

\end{proof}

\subsection{T-graphs for periodic bipartite graphs}
\label{subsec:periodicTgraphs}

The notion of T-graph was introduced in~\cite{Kenyon2004}. A pairwise disjoint collection $L_1, L_2, \dots , L_n$ 
of open line segments in $\R^2$
\emph{forms a T-graph in $\R^2$} if $\cup_{i=1}^n L_i$ is connected and contains all of its limit points except for some set $R = \{r_1, . . . , r_m\}$, where each $r_i$ lies on the boundary of the infinite component of $\R^2\smallsetminus \overline{\cup_{i=1}^n L_i}.$ Elements in $R$ are called root vertices. 
Starting from a T-graph one can define a bipartite graph, whose black vertices are the open line segments $L_i$ and
whose white vertices are the (necessarily convex) faces of the T-graph. A white vertex is adjacent to a black vertex if the  corresponding face contains a portion of the  corresponding segment as its boundary. Using a T-graph one can define in a natural geometric way (real) Kasteleyn weights on this bipartite graph: the weights are a sign $\pm$ times 
the lengths of the corresponding segments, where the sign depends on which side of the black segment the white face is on; changing the choice of which side corresponds to the $+$ sign is a gauge change.
Conversely, as described in \cite{Kenyon2004},
for a planar bipartite graph with Kasteleyn weights one can construct a T-graph 
corresponding to this bipartite graph. For infinite bi-periodic bipartite graphs one can similarly construct 
infinite T-graphs without boundary. For any $(\lambda_1,\lambda_2)$ in the liquid phase (see Section~\ref{sebsect:corresp}), we consider the realization $\phi_1:\tilde\G^* \to \mathbb{C} $ defined from $\omega_1(w,b) = \Re(G(w))K_{wb}F(b)$.

\begin{lemma}\label{tgraph} The realization $\phi_1$ defined above is a T-graph embedding.
\end{lemma}

\begin{proof} The proof starts along the lines of Theorem $4.6$ of \cite{Kenyon2004}, which deals with the finite case. 
The $\phi_1$-image
of each black face is a line segment.
For a generic direction~$\tau$, consider the inner products $\psi(f) := \phi_1(f)\cdot\tau$ as $f$ runs over vertices of $\tilde\G^*$; 
we claim that this function $\psi$ 
satisfies a maximum principle: it has no local
maxima or minima. This fact follows from the Kasteleyn matrix orientation: If a face $f$ of $\tilde\G$
has vertices $w_1,b_1,\dots,w_k,b_k$ in counterclockwise order, we denote the neighboring faces as $f_1,f_2, \dots, f_{2k}$. Then the ratios 
\[\frac{\phi_1(f)-\phi_1(f_{2i-1})}{\phi_1(f)-\phi_1(f_{2i})}=\frac{\omega_1(w_i,b_i)}{-\omega_1(w_{i+1},b_i)} = -\frac{\Re(G(w_i))K_{w_ib_i}}{\Re(G(w_{i+1}))K_{w_{i+1}b_i}}\]
(with cyclic indices)
cannot be all positive, by the Kasteleyn condition. Thus not all black faces of $\tilde \G^*$ adjacent to $f$ have $\phi_1$-image with an endpoint at $\phi_1(f)$: at least one has $\phi_1(f)$ in its interior and thus there is a neighbor of $f$ with larger value
of $\psi$ and a neighbor with smaller value of $\psi$.

It follows from Remark~\ref{rem:phiperiodic} and Proposition \ref{logslope} that $\phi_1$ is 
a locally finite realization, in the sense that any compact set contains only finitely many points of the form $\phi_1$. Indeed, in the real node case, $2\phi_1=\phi+\hat\phi$ has a periodic image of nonzero period while in the interior simple zero case, it is the sum of the periodic realization $\phi$ of nonzero period and of the realization $\hat \phi$ which is bounded by Lemma~\ref{bound_simple}.

We claim that the $\phi_1$ image of a white face $w$ is a convex polygon. If not,
we could find a vector $\tau$ and four vertices $f_1,f_2,f_3,f_4$ of $w$ in clockwise order such that both
$\psi(f_1),\psi(f_3)$ are larger than either of $\psi(f_2),\psi(f_4)$.
By the maximum principle we can then find four disjoint infinite paths starting from $f_1,f_2,f_3,f_4$ respectively 
on which $\psi(\cdot)$ is respectively increasing, decreasing, increasing, decreasing. We linearly extend $\phi_1$ in order to define it on the edges of $\tilde\G^*$. Consider a circle $\mathcal{C}$ such that the disk that it bounds contains $\phi_1(f_i)$ for all $1 \leq i \leq 4$. By the local finiteness property, this disk contains finitely many points of the realization $\phi_1$, hence the four paths must intersect~$\mathcal{C}$. Denote by $A_i$ the point at which the $i$-th path intersects $\mathcal{C}$ 
for the first time for $1 \leq i \leq 4$; these points are in cyclic order on $\mathcal{C}$, 
since the paths are disjoint. We obtain that $\psi(A_1),\psi(A_3)$ are larger than either of $\psi(A_2),\psi(A_4)$, 
which contradicts the convexity of $\mathcal{C}$ and completes the proof of the claim that the $\phi_1$-image 
of each white face is convex.

A similar argument applied to the black segments shows that the set of white faces adjacent to a black segment
winds exactly once around the black segment, rather than multiple times, so $\phi_1$ is a local embedding near a black segment.

Since $\phi_1$ is also locally finite, it has to be proper hence it is a global embedding.
\end{proof}

\subsection{Correspondence}\label{sebsect:corresp}

Let $\G$ be a bipartite graph on the torus, with an equivalence class of positive edge weights under gauge equivalence. Recall the definition of the spectral curve $P(\lambda_1,\lambda_2)=0$ defined from this data. We say that $\G$ is in a \emph{liquid phase} if the origin is in the interior of the amoeba of $P$ (this terminology comes from \cite{KOS} and refers
to the polynomial decay of correlations for the corresponding dimer model on $\tilde\G$).
In this case the roots $(\lambda_1,\lambda_2)$ of $P$ for which $|\lambda_1|=|\lambda_2|=1$ consist
in either a pair of conjugate simple 
roots $(\lambda_1,\lambda_2),(\bar\lambda_1,\bar\lambda_2)$ 
or a real node $(\lambda_1,\lambda_2)=(\pm1,\pm1)$.

By Theorem \ref{thm:periodicdimers} above, 
associated to the data of a liquid phase dimer model is a periodic, orientation preserving 
convex embedding $\phi$ of $\G^*$, well defined up to homothety and translation. 
The converse is also true, giving us a bijection between these spaces:

\begin{theorem}
\label{thm:correspondence}
For toroidal graphs, the correspondence between liquid phase dimer models and periodic circle center embeddings is bijective. 
\end{theorem}

\begin{proof}
Given a periodic, orientation-preserving embedding $\phi$ of $\tilde\G^*$ satisfying the angle condition 
(and thus a convex embedding),
we define edge weights by associating to each edge $e$ in $\tilde \G$ the complex number corresponding to the dual edge $e^*$ dual to $e$, oriented in such a way that the white dual face lies on its left. 
These edge weights define positive $X$ variables, 
because the sum of black angles equals the sum of white angles around each dual vertex. 

Let $K$ be the associated Kasteleyn matrix, with $K_{wb}$ equal to the corresponding complex edge weight.
Then we see that $K$ is in a Coulomb gauge, since the sum of the~$K_{wb}$ around each vertex is zero.

It remains to see that the weights are in a liquid phase. Since all face weights are real,
$K$ is gauge equivalent to the matrix $K(\lambda_1,\lambda_2)$, which has real weights 
except on the dual curves $l_1, l_2$,
where the weights are multiplied by $\lambda_1^{\pm1},\lambda_2^{\pm1}$ as before. Thus
$K_{wb}=G(w)K_{wb}(\lambda_1,\lambda_2)F(b)$
for some functions $G,F$. If at least one of 
$\lambda_1,\lambda_2$ is nonreal, then $(\lambda_1,\lambda_2)$ is an interior zero of $P$,
so we are in a liquid phase. If $\lambda_1,\lambda_2$ are both real,
then $K$ is real; in this case $K$ must have two dimensional kernel: both $\Re(F)$ and $\Im(F)$ are in the kernel,
and $\Re(G)$ and $\Im(G)$ are in the left kernel; since the embedding is two-dimensional either $\Re(F)$ and $\Im(F)$
are independent vectors or $\Re(G)$ and $\Im(G)$ are independent vectors.
Thus $(\lambda_1,\lambda_2)$ is at a real node of $P$ and again we are in a liquid phase.
\end{proof}

\section{Spider move, central move and Miquel dynamics}
\label{sec:Miquel}

\subsection{A central relation}

Given five distinct points $u,u_1,u_2,u_3,u_4 \in \C$, consider the following equation for an unknown $z$: 
\begin{align} \label{eq:quadeq}
	\frac{(u_2-z)(u_4-z)}{(u_1-z)(u_3-z)} = \frac{(u_2-u)(u_4-u)}{(u_1-u)(u_3-u)}.
\end{align}
This quadratic equation has two roots $z=u$ and $z=\tilde{u}$ where 
\begin{align}\label{eq:centralmove}
	\tilde{u} = \frac{u u_1 u_3 - u_1 u_2 u_3 - u u_2 u_4 + u_1 u_2 u_4 - u_1 u_3 u_4 + u_2 u_3 u_4}{
		u u_1 - u u_2 + u u_3 - u_1 u_3 - u u_4 + u_2 u_4}.
\end{align}
We call the map $u \mapsto \tilde{u}$ \emph{a central move}.

The central relation is known to be an integrable discrete equation of octahedron type. Rewriting \eqref{eq:centralmove} yields
\[
\frac{(u-u_1)(u_2-\tilde{u})(u_3-u_4)}{(u_1-u_2)(\tilde{u}-u_3)(u_4-u)} = -1
\]
which coincides with type $\chi_2$ in the Adler-Bobenko-Suris list of integrable discrete equations \cite{Adler2012} and arises in the classical Menelaus theorem in projective geometry \cite{Schief2002b}.

The central move depends on four other points $u_1,u_2,u_3,u_4$. It is related to the following circle pattern relation:
Suppose we have circles $C,C_1,C_2,C_3,C_4$ centered at $u,u_1,u_2,u_3,u_4$ with, for $i=1,2,3,4$, the triple of circles  $C,C_i,C_{i+1}$ meeting at a point (here indices are taken modulo $4$). Then, by Miquel's six circles theorem \cite{Miquel}, there is another circle $\tilde C$ which, for each $i=1,2,3,4$ 
intersects $C_i,C_{i+1}$ at the other point of intersection of $C_i$ and $C_{i+1}$. This circle $\tilde C$ has center 
$\tilde u$. 

\begin{theorem}[Centers of Miquel's six circles]
	\label{thm:Miquelcenters}
	Suppose five circles have centers $u$, $u_1$, $u_2$, $u_3$, $u_4$ as in Figure \ref{fig:Miquelcenters}. Then the center $\tilde{u}$ of the remaining circle through $\tilde{A}, \tilde{B}, \tilde{C}, \tilde{D}$ coincides with the point given by the central move \eqref{eq:centralmove}.
\end{theorem}
\begin{proof}
The existence of the sixth circle follows from Miquel's six circles theorem. 
It remains to show the relation between the centers. 
Notice that the line connecting two centers is always perpendicular the common chord of the two circles. 
The center of the sixth circle is determined uniquely as the intersection of the perpendicular bisectors of 
$\tilde{A}\tilde{B}, \tilde{B}\tilde{C},\tilde{C}\tilde{D},\tilde{D}\tilde{A}$. 
It suffices to show that the point $\tilde{u}$ determined from \eqref{eq:centralmove} lies on these perpendicular lines.
	
	We show that $\tilde{u} - u_2$ is perpendicular to $\tilde{A}-\tilde{B}$. On one hand, since $\tilde{A} ,\tilde{B}, B, A$ are concyclic, their cross ratio is real. We know that $u_3-u_2 \perp \tilde{B} - B$,  $u-u_2 \perp B - A$, $u_1-u_2 \perp A-\tilde{A}$. Thus we have $\tilde{u} - u_2 \perp \tilde{A} - \tilde{B}$ if and only if
	$
	\frac{(\tilde{u}-u_2)(u-u_2)}{(u_1-u_2)(u_3-u_2)}$ is real. On the other hand, since $u,u_1,u_2,u_3,u_4$ are circumcenters, the quantity 
	$X := - \frac{(u_2-u)(u_4-u)}{(u_1-u)(u_3-u)}$ is real, see Section~\ref{circlestoX}. Considering the quadratic equation \eqref{eq:quadeq}, the roots $u, \tilde{u}$ satisfy
	\begin{align} \label{eq:roots}
		u + \tilde{u} = \frac{(u_2+u_4) + X(u_1+u_3)}{1+X}, \quad
		u \, \tilde{u} = \frac{ u_2 u_4 + X u_1 u_3}{1+X}
	\end{align}
	and hence 
	\begin{align*}
		\frac{(\tilde{u}-u_2)(u-u_2)}{(u_1-u_2)(u_3-u_2)} = X/(1+X)
	\end{align*}
	is real. This implies $\tilde{u} - u_2$ is perpendicular to $\tilde{A}-\tilde{B}$.
	
	Similarly we can show that $\tilde{u}$ lies on the other perpendicular lines and so $\tilde{u}$ is the circumcenter of the sixth circle.
\end{proof}

\begin{figure}[h!]
	\centering
	\center{\includegraphics[width=5.2in]{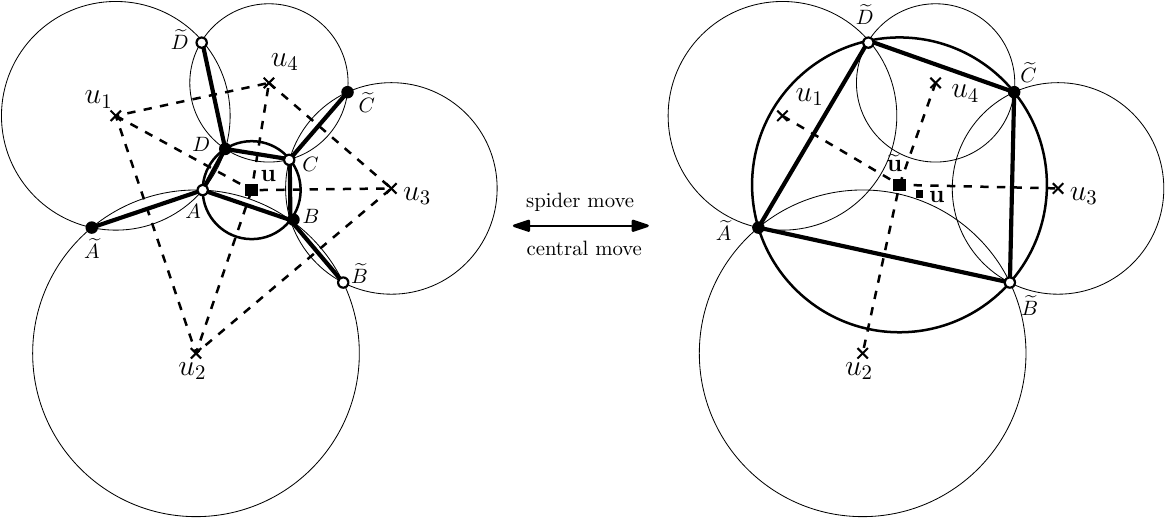}}
	\caption{Miquel's six circles theorem states that given a configuration of five circles, meeting three at a time
	at $A,B,C,D$ as indicated, there exists a sixth circle 
	passing through the four remaining intersection points. The circumcenters of the six circles satisfy the central relation~\eqref{eq:quadeq}.}
	\label{fig:Miquelcenters}
\end{figure}

\subsection{Cluster variables} \label{subsec:clustervar}

Given $u,u_1,u_2,u_3,u_4 \in \C$, we define $X=-\frac{(u_2-u)(u_4-u)}{(u_1-u)(u_3-u)}.$
More generally, if we have a bipartite circle pattern with circle centers $\{u_i\}$, and $f$ is a face of degree $2k$,
define 
$$X_f =  -\frac{(u_2-u)(u_4-u)\dots(u_{2k}-u)}{(u_1-u)(u_3-u)\dots(u_{2k-1}-u)}.$$ 
As discussed in Section \ref{circlestoX} above, this associates a real variable to every face (i.e. circle) of the circle pattern. 

\begin{theorem} \label{thm:localchange} Suppose we have a circle pattern with bipartite graph $\G$, and $f$ is a quad face of $\G$, with neighboring faces $f_1,f_2,f_3,f_4$. Let $u,u_1,u_2,u_3,u_4$ be the corresponding circle centers, and $X_1,\dots,X_4$ be the $X$ variables.
	Then under a spider move the circles undergo a Miquel transformation, the new circle center is $\tilde u$,
	and the $X$ variables are transformed as follows:
\begin{align}\nonumber X' &= X^{-1}\\
		\nonumber X_1' &= X_1(1+X)\\
	\label{1+X}	X_2'&=X_2(1+X^{-1})^{-1}\\
		\nonumber X_3' &= X_3(1+X)\\
		\nonumber X_4' &= X_4(1+X^{-1})^{-1}
\end{align}
\end{theorem}
\begin{proof}
	Equation \eqref{eq:roots} implies
	\[
	(\tilde{u} - u_i)(u - u_i)  = \frac{ u_2 u_4 + X u_1 u_3}{1+X} - u_i  \frac{(u_2+u_4) + X (u_1+u_3)}{1+X} + u_i^2.
	\]
	This factors for $i=1,2,3,4$, yielding
	\begin{align} \label{eq:chX}
		\frac{(\tilde{u} - u_i)(u - u_i)}{(u_{i+1} - u_i)(u_{i-1} - u_i)}  = \begin{cases}
			1/(1+X) \quad \text{for odd i} \\
			1/(1+X^{-1}) \quad \text{for even i.}
		\end{cases} 
	\end{align}
	Now (\ref{1+X}) is a short verification.
\end{proof}

\begin{remark}
The transformation formulas for the $X$ variables corresponds  to the transformation formulas for so-called $Y$-systems, introduced in \cite{Zamolodchikov1991} and which fall into the broader framework of coefficient variables in cluster algebras \cite{Fock2006,Fomin2007}.
\end{remark}

\subsection{Miquel dynamics}

Miquel dynamics is a dynamical system on circle patterns with the combinatorics of the square grid \cite{Ramassamy2017}. We color the faces (corresponding to circles) black and white in a chessboard fashion. 
A black mutation is to remove all the black circles and replace them by new circles obtained from Miquel's theorem, and similarly for a white mutation. More precisely, black (resp. white) mutation moves each vertex to the other intersection point of the two white (resp. black) circles it belongs to. Applying two mutations of the same type gives the identity map. 
Miquel dynamics is the process of applying black and white mutations alternately on a circle pattern. Our notion of the central move shows that the centers under Miquel dynamics follow an integrable system equivalent to that of the octahedron recurrence (defined in~\cite{Fomin2002}).

Note that, while Miquel dynamics was originally defined as a dynamics on circle patterns, it is also a well-defined dynamics on centers of circle patterns. In terms of centers $u: F(\Z^2) \to \C$, a \emph{black mutation} $M_b$ 
simply applies a central move to all black centers. In terms of spider moves on (the dual graph) $\Z^2$, a black mutation
can be decomposed into two steps: \smallskip

\textbf{Step 1:} Apply a spider move to the black faces. 

\textbf{Step 2:} Contract all the degree-2 vertices. \smallskip

The new black centers and the old white centers define a map $M_{b}(u): F(\mathbb{Z}^2) \to \mathbb{C}$ giving 
the centers of the circle pattern $M_{b}(z)$. Similarly one defines the white mutation~$M_{w}$. 
Applying black and white mutations alternately yields a sequence of square grids 

\[
\{\dots ,M_{b}(M_{w}(u)), M_{w}(u), u, M_{b}(u), M_{w}(M_{b}(u)), \dots \} 
\]

\begin{figure}[h!]
	\centering
	\begin{minipage}{0.3\textwidth}
		\begin{tikzpicture}[line cap=round,line join=round,>=triangle 45,x=1.0cm,y=1.0cm,scale=0.5]
		\clip(-4.,-4.) rectangle (4.,4.);
		\node at (1.5,1.5) {w};
		\node at (-1.5,1.5) {b};
		\node at (1.5,-1.5) {b};
		\node at (-1.5,-1.5) {w};
		\coordinate [label=below right:$z$] (A)  at (0,0);
		\draw (-3.,3.)-- (0.,3.);
		\draw (0.,3.)-- (3.,3.);
		\draw (3.,3.)-- (3.,0.);
		\draw (3.,0.)-- (3.,-3.);
		\draw (3.,-3.)-- (0.,-3.);
		\draw (0.,-3.)-- (0.,0.);
		\draw (0.,0.)-- (0.,3.);
		\draw (-3.,3.)-- (-3.,0.);
		\draw (-3.,0.)-- (0.,0.);
		\draw (-3.,-3.)-- (0.,-3.);
		\draw (-3.,0.)-- (-3.,-3.);
		\draw (0.,0.)-- (3.,0.);
		\draw (3.,0.)-- (5.,0.);
		\draw (3.,3.)-- (5.,3.);
		\draw (3.,3.)-- (3.,5.);
		\draw (0.,3.)-- (0.,5.);
		\draw (-3.,3.)-- (-3.,5.);
		\draw (-3.,3.)-- (-5.,3.);
		\draw (-3.,0.)-- (-5.,0.);
		\draw (-3.,-3.)-- (-5.,-3.);
		\draw (3.,-3.)-- (5.,-3.);
		\draw (-3.,-3.)-- (-3.,-5.);
		\draw (0.,-3.)-- (0.,-5.);
		\draw (3.,-3.)-- (3.,-5.);
		\begin{scriptsize}
		\draw [fill=black] (0.,0.) circle (1.5pt);
		\draw [fill=black] (-3.,0.) circle (1.5pt);
		\draw [fill=black] (-3.,3.) circle (1.5pt);
		\draw [fill=black] (0.,3.) circle (1.5pt);
		\draw [fill=black] (3.,3.) circle (1.5pt);
		\draw [fill=black] (3.,0.) circle (1.5pt);
		\draw [fill=black] (3.,-3.) circle (1.5pt);
		\draw [fill=black] (0.,-3.) circle (1.5pt);
		\draw [fill=black] (-3.,-3.) circle (1.5pt);
		\draw [fill=black] (5.,0.) circle (1.5pt);
		\draw [fill=black] (5.,3.) circle (1.5pt);
		\draw [fill=black] (3.,5.) circle (1.5pt);
		\draw [fill=black] (0.,5.) circle (1.5pt);
		\draw [fill=black] (-3.,5.) circle (1.5pt);
		\draw [fill=black] (-5.,3.) circle (1.5pt);
		\draw [fill=black] (-5.,0.) circle (1.5pt);
		\draw [fill=black] (-5.,-3.) circle (1.5pt);
		\draw [fill=black] (5.,-3.) circle (1.5pt);
		\draw [fill=black] (-3.,-5.) circle (1.5pt);
		\draw [fill=black] (0.,-5.) circle (1.5pt);
		\draw [fill=black] (3.,-5.) circle (1.5pt);
		\end{scriptsize}
		\end{tikzpicture}
	\end{minipage}
	\begin{minipage}{0.3\textwidth}
		\begin{tikzpicture}[line cap=round,line join=round,>=triangle 45,x=1.0cm,y=1.0cm,scale=0.5]
		\clip(-4.,-4.) rectangle (4.,4.);
		\coordinate [label=below:$\tilde{z}$] (A)  at (0,0);
		\draw (-2.,1.)-- (-1.,1.);
		\draw (-1.,1.)-- (-1.,2.);
		\draw (-1.,2.)-- (-2.,2.);
		\draw (-2.,2.)-- (-2.,1.);
		\draw (-1.,1.)-- (0.,0.);
		\draw (0.,0.)-- (1.,-1.);
		\draw (2.,-1.)-- (3.,0.);
		\draw (3.,0.)-- (4.,1.);
		\draw (4.,1.)-- (4.,2.);
		\draw (4.,2.)-- (3.,3.);
		\draw (3.,3.)-- (2.,4.);
		\draw (2.,4.)-- (1.,4.);
		\draw (1.,4.)-- (0.,3.);
		\draw (0.,3.)-- (-1.,2.);
		\draw (-2.,2.)-- (-3.,3.);
		\draw (-3.,3.)-- (-4.,4.);
		\draw (-2.,1.)-- (-3.,0.);
		\draw (-3.,0.)-- (-4.,-1.);
		\draw (-4.,-1.)-- (-4.,-2.);
		\draw (-4.,-2.)-- (-3.,-3.);
		\draw (-3.,-3.)-- (-2.,-4.);
		\draw (-2.,-4.)-- (-1.,-4.);
		\draw (-1.,-4.)-- (0.,-3.);
		\draw (0.,-3.)-- (1.,-2.);
		\draw (1.,-2.)-- (1.,-1.);
		\draw (1.,-1.)-- (2.,-1.);
		\draw (2.,-1.)-- (2.,-2.);
		\draw (2.,-2.)-- (1.,-2.);
		\draw (2.,-2.)-- (3.,-3.);
		\draw (3.,-3.)-- (4.,-4.);
		\begin{scriptsize}
		\draw [fill=black] (0.,0.) circle (1.5pt);
		\draw [fill=black] (-3.,0.) circle (1.5pt);
		\draw [fill=black] (-3.,3.) circle (1.5pt);
		\draw [fill=black] (0.,3.) circle (1.5pt);
		\draw [fill=black] (3.,3.) circle (1.5pt);
		\draw [fill=black] (3.,0.) circle (1.5pt);
		\draw [fill=black] (3.,-3.) circle (1.5pt);
		\draw [fill=black] (0.,-3.) circle (1.5pt);
		\draw [fill=black] (-3.,-3.) circle (1.5pt);
		\draw [fill=black] (5.,0.) circle (1.5pt);
		\draw [fill=black] (5.,3.) circle (1.5pt);
		\draw [fill=black] (3.,5.) circle (1.5pt);
		\draw [fill=black] (0.,5.) circle (1.5pt);
		\draw [fill=black] (-3.,5.) circle (1.5pt);
		\draw [fill=black] (-5.,3.) circle (1.5pt);
		\draw [fill=black] (-5.,0.) circle (1.5pt);
		\draw [fill=black] (-5.,-3.) circle (1.5pt);
		\draw [fill=black] (5.,-3.) circle (1.5pt);
		\draw [fill=black] (-3.,-5.) circle (1.5pt);
		\draw [fill=black] (0.,-5.) circle (1.5pt);
		\draw [fill=black] (3.,-5.) circle (1.5pt);
		\draw [fill=black] (-2.,2.) circle (1.5pt);
		\draw [fill=black] (-2.,1.) circle (1.5pt);
		\draw [fill=black] (-1.,1.) circle (1.5pt);
		\draw [fill=black] (-1.,2.) circle (1.5pt);
		\draw [fill=black] (1.,-1.) circle (1.5pt);
		\draw [fill=black] (1.,-2.) circle (1.5pt);
		\draw [fill=black] (2.,-2.) circle (1.5pt);
		\draw [fill=black] (2.,-1.) circle (1.5pt);
		\draw [fill=black] (-4.,4.) circle (1.5pt);
		\draw [fill=black] (1.,4.) circle (1.5pt);
		\draw [fill=black] (2.,4.) circle (1.5pt);
		\draw [fill=black] (4.,2.) circle (1.5pt);
		\draw [fill=black] (4.,1.) circle (1.5pt);
		\draw [fill=black] (-2.,-4.) circle (1.5pt);
		\draw [fill=black] (-1.,-4.) circle (1.5pt);
		\draw [fill=black] (-4.,-1.) circle (1.5pt);
		\draw [fill=black] (-4.,-2.) circle (1.5pt);
		\draw [fill=black] (4.,-4.) circle (1.5pt);
		\end{scriptsize}
		\end{tikzpicture}
	\end{minipage}
	\begin{minipage}{0.3\textwidth}
		\begin{tikzpicture}[line cap=round,line join=round,>=triangle 45,x=1.0cm,y=1.0cm,scale=0.5]
		\clip(-4.,-4.) rectangle (4.,4.);
		\coordinate [label=below right:$M_{b}(z)$] (A)  at (0,0);
		\draw (-3.,3.)-- (0.,3.);
		\draw (0.,3.)-- (3.,3.);
		\draw (3.,3.)-- (3.,0.);
		\draw (3.,0.)-- (3.,-3.);
		\draw (3.,-3.)-- (0.,-3.);
		\draw (0.,-3.)-- (0.,0.);
		\draw (0.,0.)-- (0.,3.);
		\draw (-3.,3.)-- (-3.,0.);
		\draw (-3.,0.)-- (0.,0.);
		\draw (-3.,-3.)-- (0.,-3.);
		\draw (-3.,0.)-- (-3.,-3.);
		\draw (0.,0.)-- (3.,0.);
		\draw (3.,0.)-- (5.,0.);
		\draw (3.,3.)-- (5.,3.);
		\draw (3.,3.)-- (3.,5.);
		\draw (0.,3.)-- (0.,5.);
		\draw (-3.,3.)-- (-3.,5.);
		\draw (-3.,3.)-- (-5.,3.);
		\draw (-3.,0.)-- (-5.,0.);
		\draw (-3.,-3.)-- (-5.,-3.);
		\draw (3.,-3.)-- (5.,-3.);
		\draw (-3.,-3.)-- (-3.,-5.);
		\draw (0.,-3.)-- (0.,-5.);
		\draw (3.,-3.)-- (3.,-5.);
		\begin{scriptsize}
		\draw [fill=black] (0.,0.) circle (1.5pt);
		\draw [fill=black] (-3.,0.) circle (1.5pt);
		\draw [fill=black] (-3.,3.) circle (1.5pt);
		\draw [fill=black] (0.,3.) circle (1.5pt);
		\draw [fill=black] (3.,3.) circle (1.5pt);
		\draw [fill=black] (3.,0.) circle (1.5pt);
		\draw [fill=black] (3.,-3.) circle (1.5pt);
		\draw [fill=black] (0.,-3.) circle (1.5pt);
		\draw [fill=black] (-3.,-3.) circle (1.5pt);
		\draw [fill=black] (5.,0.) circle (1.5pt);
		\draw [fill=black] (5.,3.) circle (1.5pt);
		\draw [fill=black] (3.,5.) circle (1.5pt);
		\draw [fill=black] (0.,5.) circle (1.5pt);
		\draw [fill=black] (-3.,5.) circle (1.5pt);
		\draw [fill=black] (-5.,3.) circle (1.5pt);
		\draw [fill=black] (-5.,0.) circle (1.5pt);
		\draw [fill=black] (-5.,-3.) circle (1.5pt);
		\draw [fill=black] (5.,-3.) circle (1.5pt);
		\draw [fill=black] (-3.,-5.) circle (1.5pt);
		\draw [fill=black] (0.,-5.) circle (1.5pt);
		\draw [fill=black] (3.,-5.) circle (1.5pt);
		\end{scriptsize}
		\end{tikzpicture}
	\end{minipage}
	\caption{Starting with $z:\mathbb{Z}^2\to \mathbb{C}$ with black and white faces (Left), a black mutation $M_b$ applied a spider move to each black face (Center). We then contract all the degree-2 vertices and obtain a new immersion $M_{b}(z): \mathbb{Z}^2 \to \mathbb{C}$ (Right). }
	\label{fig:miqueldyn}
\end{figure}

As in section \ref{subsec:clustervar}, a weight $X: F(\mathbb{Z}^2) \to \R_{>0}$ is associated to the centers of the circle pattern.  Let $M_{b}( X_{m,n})$ be a weight of face $(m,n)$ after a mutation $M_b$.

\begin{proposition}\label{prop:changeofX}
	Under a black mutation
	\begin{align*}
	M_{b}( X_{m,n}) =\begin{cases}
	X^{-1}_{m,n} \quad &\text{if $(m,n)$ is a black face} \\
	X_{m,n} \frac{(1+X_{m,n+1})(1+X_{m,n-1})}{(1+X^{-1}_{m+1,n})(1+X^{-1}_{m-1,n})} \quad &\text{if $(m,n)$ is a white face} \\
	\end{cases}
	\end{align*}
	In particular, $X_{m,n} > 0$ for all $(m,n)\in\Z^2$ if and only if $M_b(X_{m,n}) >0$ for all $(m,n)\in\Z^2$. The same holds for a white mutation.
\end{proposition}
\begin{proof}
The formulas follow from Theorem \ref{thm:localchange}. If $X_{m,n}>0$ for all $(m,n)\in\Z^2$, the equations yield $M_b(X_{m,n})>0$ immediately. If $M_b(X_{m,n}) >0$ for all $(m,n)\in\Z^2$, then 
$X_{m,n}=M_b(M_b (X_{m,n})) >0$.
\end{proof}

This implies that the class of circle patterns with positive face weights is preserved under Miquel dynamics. 

The class of spatially biperiodic circle patterns is also preserved by Miquel dynamics \cite{Ramassamy2017} and in that case, Miquel dynamics is integrable and one can deduce a complete set of 
invariants from the partition function of the underlying dimer model \cite{GK2013}.

\subsection{Fixed points of Miquel dynamics}

In this subsection we study the fixed points of Miquel dynamics. Given a circle pattern, one can construct its centers and given the centers, one can compute the $X$ variables. This gives three possible definitions of ``fixed point of Miquel dynamics", 
in increasing order of strength: either the collection of $X$ variables is preserved, or 
the collection of centers is preserved, or the circle pattern itself is preserved.

We first consider the case of a center being fixed by the central move. We rewrite the central move (\ref{eq:centralmove}) as follows: 
\begin{lemma}\label{lemma:quadcoeff} 
	Suppose $X:= - \frac{(u_1-u)(u_3-u)}{(u_2-u)(u_4-u)} >0$.
	If $u_1-u_2+u_3-u_4 \neq 0$, we have
	\[
	\tilde{u} = u + \frac{((u_1-u)(u_3-u)-(u_2-u)(u_4-u))^2-(u_1-u_2)(u_2-u_3)(u_3-u_4)(u_4-u_1)}{((u_1-u)(u_3-u)-(u_2-u)(u_4-u))(u_1-u_2+u_3-u_4)}
	\]
	If $u_1-u_2+u_3-u_4 = 0$, we have
	\[
	\tilde{u}= u + \frac{(u_1-u)+(u_2-u)+(u_3-u)+(u_4-u)}{2}.
	\]
\end{lemma}

Recall that a tangential quadrilateral is a quadrilateral with an incircle, i.e. a circle tangent to the extended lines of the four sides. The incircle is unique if it exists. In this case, the center of the incircle is the intersection of bisectors of interior angles at the four corners.

\begin{proposition}\label{prop:localfix}
	Suppose $X= - \frac{(u_2-u)(u_4-u)}{(u_1-u)(u_3-u)} > 0$. Then the quadratic equation \eqref{eq:quadeq} has a repeated root $u = \tilde{u}$ if and only if $u_1 u_2 u_3 u_4$ forms a tangential quadrilateral with an incircle centered at $u$.  
\end{proposition}
\begin{proof}
	If $u_1-u_2+u_3-u_4 =0$, then setting $\tilde{u}=u$ in Lemma \ref{lemma:quadcoeff} implies 
	  \[
	u= \frac{u_1+u_3}{2} = \frac{u_2+u_4}{2}.
	\]
	is the center of the parallelogram $u_1 u_2 u_3 u_4$. Furthermore $X$ being positive implies the parallelogram is a rhombus whose incircle is centered at $u$.

   Otherwise, we assume $u_1-u_2+u_3-u_4 \neq 0$. Lemma \ref{lemma:quadcoeff} implies the segment $u\, u_i$ is an angle bisector of $\angle u_{i-1}u_i u_{i+1}$. To see this, setting $\tilde{u}=u$ in Lemma  \ref{lemma:quadcoeff} yields
   \[
   u_4 = -\frac{(u - u_1) (u - u_2) (u - u_3)}{ (u-u_2)^2 - (u_3-u_2)(u_1-u_2)} + u
   \]
   Substituting it into $X$, we have
   \[
   X= - \frac{(u_2-u)(u_4-u)}{(u_1-u)(u_3-u)} = \frac{(u-u_2)^2}{-(u-u_2)^2 + (u_3-u_2)(u_1-u_2)}.
   \]
   Since $X$ is positive, we can deduce $\frac{(u_3-u_2)(u_1-u_2)}{(u-u_2)^2}=\frac1X+1$ is positive as well and thus the segment $u\, u_2$ is an angle bisector of $\angle u_1 u_2 u_3$. Similarly we can deduce that $u$ lies on the angles bisectors of the other three corners and hence there is an incircle tangent to $u_1 u_2 u_3 u_4$ and centered at $u$. 
\end{proof}

We can now characterize the centers that are preserved under the Miquel dynamics.

\begin{theorem}
\label{thm:fixedcenters}
The centers of a circle pattern 
are preserved under Miquel dynamics if and only if they are also the centers of some circle pattern where diagonal circles are tangential, i.e. the circles centered at $u_{m,n}$ and $u_{m+1,n+1}$ are tangential, and the circles centered
at $u_{m,n}$ and $u_{m-1,n+1}$ are tangential. 
\end{theorem}
\begin{proof}
	Suppose the centers $u:F(\mathbb{Z}^2) \to \mathbb{C}$ are fixed under the central move. We define $z: \mathbb{Z}^2 \to \mathbb{C}$ to be the intersection of the diagonals in each elementary 
	quadrilateral  $u_{m,n} u_{m+1,n+1} \cap u_{m+1,n} u_{m,n+1}$. 
	We claim the faces of $z$ are cyclic, centered at the $u$'s.  
	To see this it suffices to show that $z_r := u_{m,n} u_{m+1,n+1} \cap  u_{m,n+1} u_{m+1,n}$ is the image of $z_l := u_{m,n} u_{m-1,n+1} \cap u_{m,n+1} u_{m-1,n} $ under the reflection across $u_{m,n} u_{m,n+1}$ (See Figure \ref{fig:inquad}). Indeed, Proposition  \ref{prop:localfix} implies that $u_{m,n}$ is at the center of the inscribed circle of the quadrilateral
	$u_{m+1,n} u_{m,n+1} u_{m-1,n} u_{m,n-1}$. It yields that under the reflection the 
	ray through $u_{m,n+1}$ and $u_{m+1,n}$ is the image of the 
	ray through $u_{m,n+1}$ and $u_{m-1,n}$. By considering a nearby quadrilateral similarly, we can deduce the 
	ray through from $u_{m,n}$ and $u_{m+1,n+1}$ is the image of the 
	ray through $u_{m,n}$ and $u_{m-1,n+1}$.
	
Now consider the circles defined by the faces of the $z$'s; this is a \emph{new} set of circles centered at the $u$'s
(not those inscribed in the quadrilaterals). 
We claim that the diagonal circles are tangent to each other. To see this, consider the point $z_r$ which is the
intersection of the diagonals $u_{m,n} u_{m+1,n+1} \cap u_{m+1,n} u_{m,n+1}$. 
Notice that the distance between $z_r$ and $u_{m,n}$ is the radius of the circle at $u_{m,n}$ and 
similarly the distance between $z_r$ and $u_{m+1,n+1}$ is the radius of the circle at $u_{m+1,n+1}$. 
Since $z_r$ lies on the line joining $u_{m,n}$ and $u_{m+1,n+1}$, which is perpendicular to both circles, 
the opposite circles are tangential.
\end{proof}
\begin{figure}[h!]
\begin{tikzpicture}[line cap=round,line join=round,>=triangle 45,x=1.0cm,y=1.0cm]
\clip(-0.8,-0.1) rectangle (5.7,4.8);
\draw [line width=2.pt,dash pattern=on 5pt off 5pt] (3.,2.)-- (5.,4.);
\draw [line width=2.pt,dash pattern=on 5pt off 5pt] (3.,4.)-- (5.,2.);
\draw [line width=2.pt,dash pattern=on 5pt off 5pt] (3.,2.)-- (1.,4.);
\draw [line width=2.pt,dash pattern=on 5pt off 5pt] (3.,4.)-- (1.,2.);
\draw (2.6099248083708937,2.0159671446422041) node[anchor=north west] {$u_{m,n}$};
\draw (2.29051756905672246,0.95118090910435026) node[anchor=north west] {$u_{m,n-1}$};
\draw (3.50764130907734184,2.306236263982005) node[anchor=north west] {$u_{m+1,n}$};
\draw (2.29051756905672246,4.670586901763068) node[anchor=north west] {$u_{m,n+1}$};
\draw (4.033435932658721,4.670586901763068) node[anchor=north west] {$u_{m+1,n+1}$};
\draw (0.0119099759871242,4.670586901763068) node[anchor=north west] {$u_{m-1,n+1}$};
\draw (1.1061846719203975,2.306236263982005) node[anchor=north west] {$u_{m-1,n}$};
\draw (4.1018328475200375,3.3250146680112197) node[anchor=north west] {$z_r$};
\draw (1.3434365259683686,3.3250146680112197) node[anchor=north west] {$z_l$};
\draw (4.033435932658721,0.5584504437609059) node[anchor=north west] {$u_{m+1,n-1}$};
\draw (0.0119099759871242,0.5584504437609059) node[anchor=north west] {$u_{m-1,n-1}$};
\draw [line width=2.pt,dash pattern=on 5pt off 5pt] (5.,2.)-- (3.,0.);
\draw [line width=2.pt,dash pattern=on 5pt off 5pt] (3.,0.)-- (1.,2.);
\begin{scriptsize}
\draw [fill=black] (3.,2.) circle (2.0pt);
\draw [fill=black] (5.,4.) circle (2.0pt);
\draw [fill=black] (5.,2.) circle (2.0pt);
\draw [fill=black] (3.,4.) circle (2.0pt);
\draw [fill=black] (1.,4.) circle (2.0pt);
\draw [fill=black] (1.,2.) circle (2.0pt);
\draw [fill=black] (3.,0.) circle (2.0pt);
\draw [fill=black] (2.,3.) circle (2.0pt);
\draw [fill=black] (4.,3.) circle (2.0pt);
\draw [fill=black] (5.,0.) circle (2.0pt);
\draw [fill=black] (1.,0.) circle (2.0pt);
\end{scriptsize}
\end{tikzpicture}
	\caption{$z_r$ is the reflection of $z_l$ with respect to the ray through $u_{m,n}$ and $u_{m,n+1}$.}\label{fig:inquad}
\end{figure}

A particular case where opposite circles are tangential is the case of 
circle patterns with constant intersection angles. For a circle pattern $z: V(\mathbb{Z}^2) \to \mathbb{C}$, one can measure the intersection angle $\theta: E(\mathbb{Z}^2) \to [0, \pi]$ between neighboring circles. We say a circle pattern has \emph{constant horizontal and vertical intersection angles} if there exists $\alpha \in [0, \pi]$ such that
\[
\theta = \begin{cases}
\alpha &\quad \text{along vertical edges} \\
\pi - \alpha  &\quad \text{along horizontal edges}
\end{cases}
\]
It is an \emph{orthogonal circle pattern} if $\alpha = \pi/2$, see \cite{Schramm1997}. An example of a circle pattern with constant intersection angles $\alpha$ is that of a regular rectangular grid with rectangles of width $1$ and height $\cot\alpha$.

If a circle pattern has the same centers as that of a circle pattern with constant intersection angles, then its intersection points undergo rigid rotation under Miquel dynamics:
\begin{corollary}
Suppose $u:\Z^2\to\C$ gives the centers of a circle pattern with constant intersection angle $\alpha$ and $z$ is some 
other circle pattern with $u$ as centers. Then the orbit of every intersection point $z$ of the circle pattern lies on a circle,
and $M_b\circ M_w$ rotates the point around the circle by angle $2\alpha$.
\end{corollary}
\begin{proof}
	For each elementary quad $Q=u_{m,n}u_{m+1,n}u_{m+1,n+1}u_{m,n+1}$, the diagonals intersect at angle $\alpha$. We denote $\tilde{z}_0$ the intersection of the diagonals while $z_0$ is the intersection of the circles at $u_{m,n}u_{m+1,n}u_{m+1,n+1}u_{m,n+1}$. Generally, $z_0 \neq \tilde{z}_0$ unless $z$ is the circle pattern of constant intersection angle. Applying a Miquel's move $M_w$ once, $\tilde z_0$ is fixed while $z_0$ is reflected across one of the diagonals to $M_{w}(z_0)$. Applying a black mutation $M_b$, the point is reflected across the other diagonal. In both cases, the distance to $\tilde{z}$ is preserved. Hence the orbit of $z_0$ lies on a circle centered at $\tilde{z}_0$.
	
Thus $z_0$ is reflected successively across two lines (emanating from $\tilde z_0$) meeting at angle $\alpha$; two such reflections define a rotation
around $\tilde z_0$ of angle $2\alpha$.
\end{proof}

\begin{corollary}
	A circle pattern is preserved under Miquel dynamics if and only if it has constant horizontal and vertical intersection angles.
\end{corollary}

In the infinite planar case, the intersection angles between neighboring circles do not determine the pattern, hence for each value of $\alpha\in[0,\pi]$, there exists a large class of circle patterns with constant intersection angles equal to $\alpha$ for vertical edges. This class was studied in \cite{Schramm1997} in the particular case when $\alpha=\tfrac{\pi}{2}$.
In the case of spatially biperiodic patterns of prescribed periods, or equivalently, circle patterns on a flat torus, 
the intersection angles do characterize the circle pattern up to similarity \cite{Bobenko2004}, 
so that the only spatially biperiodic diagonally tangent patterns are those corresponding to 
regular rectangular grids (all the columns have the same width and all the rows have the same height). 
Hence in the spatially biperiodic case, the regular rectangular grids are the only fixed points of Miquel dynamics, 
seen either as a dynamics on circle patterns or on circle centers. The $X$ variables of a given such fixed point are the same for all the faces (equal to the squared aspect ratio of the rectangle formed by a face).

For general, not necessarily biperiodic patterns we have the following.

\begin{proposition}
	Suppose a circle pattern has face weights $X:F(\mathbb{Z}^2) \to \mathbb{R}_{>0}$. If the centers are preserved under Miquel dynamics, then
	\begin{equation} \label{eq:fixX}
	X^2_{m,n} = \frac{(1+X^{-1}_{m+1,n})(1+X^{-1}_{m-1,n})}{(1+X_{m,n+1})(1+X_{m,n-1})}.
	\end{equation}
\end{proposition}
\begin{proof}
	This follows from Proposition \ref{prop:changeofX} and the fact that the labeling of the black and white vertices switches.
\end{proof}

We showed that the centers of a circle pattern with constant intersection angles are fixed by the central moves and hence their $X$ variables satisfy Equation \eqref{eq:fixX}. The converse might not be true. For example the rectangular grid with even (resp. odd) columns of width $1$ (resp. $2$) and even (resp. odd) rows of height $1$ (resp. $2$) has all $X$ variables equal to $1$, hence satisfies Equation \eqref{eq:fixX}, but it is not diagonally tangent hence not a fixed point of Miquel dynamics. It would be interesting to find all such examples. 
\medskip

\noindent{\bf Question:} Characterize circles patterns with $X$ variables satisfying  Equation \eqref{eq:fixX}.

\subsection{Integrals of motion for Miquel dynamics}

Miquel dynamics seen as a dynamics on circle centers on an $m$ by $n$ square grid with $m$ and $n$ even
 on the torus corresponds to the dimer urban renewal dynamics on the same graph, which is a finite-dimensional integrable system \cite{GK2013}. It follows from \cite{GK2013} that the spectral curve of the dimer model associated with the successive collections of circle centers is kept invariant by Miquel dynamics. 
The integrals of motion of the dimer dynamics have an interpretation in terms of partition functions for dimer configurations with a prescribed homology (the coefficients of the polynomial used to define the spectral curve) and it would be interesting to find a geometric interpretation (in terms of circle patterns) of all these integrals of motion.

It was shown in \cite{Ramassamy2017} that the sum along any zigzag loop of intersection angles of circles is an integral of motion. This sum can actually be rewritten as the sum of the turning angles along a dual zigzag loop, which is equal to twice the argument of the alternating product along a primal zigzag loop of the associated complex edge weights. In Section \ref{sectionbiperiodic}, we associated to each collection of circle centers on the torus a point on the spectral curve of the associated dimer model. It follows from Theorem~\ref{thm:correspondence} and the conservation of the sum of angles along zigzag loops that this point on the spectral curve is kept invariant under Miquel dynamics.

\section{From planar networks to circle patterns}
\label{sec:networks}

\subsection{Harmonic embeddings of planar networks}

A \emph{circular planar network} is an embedded planar graph $\G=(V(\G), E(\G),F(\G))$, with a distinguished
subset ${B\subset V(\G)}$
of vertices on the outer face called boundary vertices, and with a conductance function $c:{ E(\G)}\to\R_{>0}$ on edges. 
Associated to this data is a Laplacian operator
$\Delta:\C^{{ V(\G)}}\to\C^{{ V(\G)}}$ defined by
$$\Delta h(v)  = \sum_{w\sim v} c_{vw}(h(v)-h(w)).$$
An embedding $h:{ V(\G)}\to\C$ is \emph{harmonic} if $\Delta h(v)=0$ for $v\in { V(\G)}\setminus B$. 
Harmonic embedding of circular planar networks arise in 
various contexts, e.g. resistor networks, equilibrium stress
configurations, and random walks \cite{Biskup11}.

Let $\G$ be a circular planar network, with boundary consisting of
vertices $B=\{v_1,\dots,v_k\}$ on the outer face. We define an augmented dual $\G^*$ to $\G$ in a similar way as for bipartite graphs, except that the intermediate
graph $\hat\G$ (whose dual is $\G^*$) has an edge to $v_\infty$ only from each boundary vertex of $\G$. 
Thus $\G^*$ is also a circular
planar network with boundary $B^*$ consisting of $k=|B|$ vertices, one between each pair of boundary vertices of $\G$.
Let $P$ be a convex $k$-gon with vertices $z_1,\dots,z_k$.
One can find a function $z:{V(\G)}\to\C$ harmonic on ${ V(\G)}\setminus B$ and with
values $z_i$ at $v_i$ for $i=1,\dots,k$. Then $z$ defines a harmonic embedding of $\G$, also known as the Tutte embedding, see~\cite{Tutte}.

We can also define a harmonic embedding of the dual graph $\G^*$ (harmonic on $\G^*\setminus B^*$)
as follows. If $v$ and $v'$ are two primal vertices and $f$ (resp. $f'$) denotes the 
dual vertex associated with the face to the right (resp. left) of the edge $vv'$ when traversed from $v$ to $v'$, then we set
\be\label{zzzz}
z(f')- z(f)=ic_{vv'}(z(v')-z(v)).
\ee
Since the function $z$ is harmonic, this defines a unique embedding of the dual $\G^*$
once one fixes the position of a single dual vertex. This embedding of the dual graph is also harmonic with respect to the inverse conductance (one should take $c_{ff'}=c_{vv'}^{-1}$). Each primal edge is orthogonal to its corresponding dual edge, hence the pair constituted of the harmonic embeddings of the primal and the dual graph form a pair of so-called \emph{reciprocal figures}. Note that given the harmonic embedding one can reconstruct the conductances from (\ref{zzzz}).

\subsection{From harmonic embeddings to circle patterns}

There is a map, known as Temperley's bijection \cite{Temperley1974,Kenyon2000}, from a circular planar network $\G$ to a face-weighted bipartite graph $\GD$, defined as follows. To every vertex and every face of $\G$ is associated a black vertex of $\GD$. 
To every edge of $\G$ is associated a white vertex of $\GD$. 
A white vertex and a black vertex of $\GD$ are connected if the corresponding edge in $\G$ 
is adjacent to the corresponding vertex or face in $\G$. Every bounded face of $\GD$ is a quadrilateral 
consisting of two white vertices and two black vertices as in the middle of Figure \ref{fig:powerdiagramtocirclepattern}. 
The bipartite graph has face weights
\[
X_f = c_{e_1}/c_{e_2}
\]
where $e_1,e_2$ are two consecutive edges of $\G$ adjacent to face $f$ of $\GD$. 
 
For
these weights the partition function of the planar network on $\G$ is equal to the partition function of the 
dimer model on $\GD$ up to a multiplicative constant \cite{Kenyon2000}. 
\begin{figure}[h!]
	\centering
	\begin{minipage}{0.32\textwidth}
		\definecolor{ffffff}{rgb}{1.,1.,1.}
		\begin{tikzpicture}[line cap=round,line join=round,>=triangle 45,x=1.0cm,y=1.0cm]
		\clip(-1.8,-0.45) rectangle (2.5,3.6);
		\draw [line width=0.4pt] (0.28,1.06)-- (-1.5047068281746292,2.0662923329465546);
		\draw [line width=0.4pt] (0.28,1.06)-- (2.111200236543201,2.177618701177263);
		\draw [line width=0.4pt] (0.28,1.06)-- (0.43249437150906295,-0.3323303513779528);
		\draw [line width=0.4pt] (-1.5552453216830588,0.025081322874544384)-- (2.129657248435265,0.4286686618607269);
		\draw [line width=0.4pt] (2.129657248435265,0.4286686618607269)-- (0.332557249046033,3.3731880587774703);
		\draw [line width=0.4pt] (0.332557249046033,3.3731880587774703)-- (-1.5552453216830588,0.025081322874544384);
		\draw [line width=0.4pt] (2.129657248435265,0.4286686618607269)-- (1.8945753048300835,-1.9466207653957424);
		\draw [line width=0.4pt] (2.129657248435265,0.4286686618607269)-- (4.8309293258441,-0.4721425660024394);
		\draw [line width=0.4pt] (2.129657248435265,0.4286686618607269)-- (4.301629459395222,2.0861401218338043);
		\draw [line width=0.4pt] (0.332557249046033,3.3731880587774703)-- (2.8019465044567324,4.064713432130801);
		\draw [line width=0.4pt] (0.332557249046033,3.3731880587774703)-- (-1.659295227040958,3.863075387769323);
		\draw [line width=0.4pt] (-1.5552453216830588,0.025081322874544384)-- (-3.3984233596587017,1.2039736777523409);
		\draw [line width=0.4pt] (-1.5552453216830588,0.025081322874544384)-- (-0.852743049595048,-1.6945732099438957);
		\begin{scriptsize}
		\draw [fill=black] (0.28,1.06) circle (2.0pt);
		\draw [fill=black] (-1.5047068281746292,2.0662923329465546) circle (2.0pt);
		\draw [fill=black] (2.111200236543201,2.177618701177263) circle (2.0pt);
		\draw [fill=black] (0.43249437150906295,-0.3323303513779528) circle (2.0pt);
		\draw [fill=ffffff] (-0.6697671999640902,1.5955184595699121) circle (2.0pt);
		\draw [fill=ffffff] (1.3469235952604774,1.7111651423993575) circle (2.0pt);
		\draw [fill=ffffff] (0.370251468204276,0.235970514887052) circle (2.0pt);
		\draw [fill=black] (0.332557249046033,3.3731880587774703) circle (2.0pt);
		\draw [fill=black] (2.129657248435265,0.4286686618607269) circle (2.0pt);
		\draw [fill=black] (-1.5552453216830588,0.025081322874544384) circle (2.0pt);
		\draw [fill=black] (1.8945753048300835,-1.9466207653957424) circle (2.0pt);
		\draw [fill=black] (4.8309293258441,-0.4721425660024394) circle (2.0pt);
		\draw [fill=black] (4.301629459395222,2.0861401218338043) circle (2.0pt);
		\draw [fill=black] (2.8019465044567324,4.064713432130801) circle (2.0pt);
		\draw [fill=black] (-1.659295227040958,3.863075387769323) circle (2.0pt);
		\draw [fill=black] (-3.3984233596587017,1.2039736777523409) circle (2.0pt);
		\draw [fill=black] (-0.852743049595048,-1.6945732099438957) circle (2.0pt);
		\end{scriptsize}
		\end{tikzpicture}
	\end{minipage}
	\begin{minipage}{0.32\textwidth}
		\definecolor{ffffff}{rgb}{1.,1.,1.}
		\begin{tikzpicture}[line cap=round,line join=round,>=triangle 45,x=1.0cm,y=1.0cm]
		\clip(-1.8,-0.45) rectangle (2.5,3.6);
		\draw [line width=0.4pt] (0.28,1.06)-- (-1.5047068281746292,2.0662923329465546);
		\draw [line width=0.4pt] (0.28,1.06)-- (2.111200236543201,2.177618701177263);
		\draw [line width=0.4pt] (0.28,1.06)-- (0.43249437150906295,-0.3323303513779528);
		\draw [line width=0.4pt] (-1.5552453216830588,0.025081322874544384)-- (2.129657248435265,0.4286686618607269);
		\draw [line width=0.4pt] (2.129657248435265,0.4286686618607269)-- (0.332557249046033,3.3731880587774703);
		\draw [line width=0.4pt] (0.332557249046033,3.3731880587774703)-- (-1.5552453216830588,0.025081322874544384);
		\draw [line width=0.4pt] (2.129657248435265,0.4286686618607269)-- (1.8945753048300835,-1.9466207653957424);
		\draw [line width=0.4pt] (2.129657248435265,0.4286686618607269)-- (4.8309293258441,-0.4721425660024394);
		\draw [line width=0.4pt] (2.129657248435265,0.4286686618607269)-- (4.301629459395222,2.0861401218338043);
		\draw [line width=0.4pt] (0.332557249046033,3.3731880587774703)-- (2.8019465044567324,4.064713432130801);
		\draw [line width=0.4pt] (0.332557249046033,3.3731880587774703)-- (-1.659295227040958,3.863075387769323);
		\draw [line width=0.4pt] (-1.5552453216830588,0.025081322874544384)-- (-3.3984233596587017,1.2039736777523409);
		\draw [line width=0.4pt] (-1.5552453216830588,0.025081322874544384)-- (-0.852743049595048,-1.6945732099438957);
		\draw [line width=0.8pt] (0.3062786245230165,2.216594029388735) circle (1.1568925252263027cm);
		\draw [line width=0.8pt] (-0.6376226608415293,0.5425406614372723) circle (1.053468326413122cm);
		\draw [line width=0.8pt] (1.2048286242176325,0.7443343309303635) circle (0.9772168637520846cm);
		\draw [line width=0.8pt] (-1.529976074928844,1.0456868279105496) circle (1.0209182786792117cm);
		\draw [line width=0.8pt] (-0.5860747895642985,2.719740195862012) circle (1.1273326624870545cm);
		\draw [line width=0.8pt] (1.221878742794617,2.7754033799773667) circle (1.0715592570881722cm);
		\draw [line width=0.8pt] (2.1204287424892336,1.303143681518995) circle (0.8745237134168086cm);
		\draw [line width=0.8pt] (1.2810758099721642,0.048169155241387054) circle (0.929984049455509cm);
		\draw [line width=0.8pt] (-0.561375475086998,-0.15362451425170423) circle (1.009808421531466cm);
		\begin{scriptsize}
		\draw [fill=black] (0.28,1.06) circle (2.0pt);
		\draw [fill=black] (-1.5047068281746292,2.0662923329465546) circle (2.0pt);
		\draw [fill=black] (2.111200236543201,2.177618701177263) circle (2.0pt);
		\draw [fill=black] (0.43249437150906295,-0.3323303513779528) circle (2.0pt);
		\draw [fill=ffffff] (-0.6697671999640902,1.5955184595699121) circle (2.0pt);
		\draw [fill=ffffff] (1.3469235952604774,1.7111651423993575) circle (2.0pt);
		\draw [fill=ffffff] (0.370251468204276,0.235970514887052) circle (2.0pt);
		\draw [fill=black] (0.332557249046033,3.3731880587774703) circle (2.0pt);
		\draw [fill=black] (2.129657248435265,0.4286686618607269) circle (2.0pt);
		\draw [fill=black] (-1.5552453216830588,0.025081322874544384) circle (2.0pt);
		\draw [fill=black] (1.8945753048300835,-1.9466207653957424) circle (2.0pt);
		\draw [fill=black] (4.8309293258441,-0.4721425660024394) circle (2.0pt);
		\draw [fill=black] (4.301629459395222,2.0861401218338043) circle (2.0pt);
		\draw [fill=black] (2.8019465044567324,4.064713432130801) circle (2.0pt);
		\draw [fill=black] (-1.659295227040958,3.863075387769323) circle (2.0pt);
		\draw [fill=black] (-3.3984233596587017,1.2039736777523409) circle (2.0pt);
		\draw [fill=black] (-0.852743049595048,-1.6945732099438957) circle (2.0pt);
		\end{scriptsize}
		\end{tikzpicture}
	\end{minipage}
	
	\caption{From the vertices of a reciprocal figure to a circle pattern.}
	\label{fig:powerdiagramtocirclepattern}
\end{figure}

In this section we convert a reciprocal figure into a circle pattern in such a way that the following diagram commutes:
$$\begin{array}{ccc}
\text{Planar network $\G$}&\longrightarrow&\text{Bipartite graph $\GD$}\\
\updownarrow&&\updownarrow\\
\text{Reciprocal figure}&\longrightarrow&\text{Circle pattern}
\end{array}
$$

\begin{theorem}
\label{thm:powerdiagramtocirclepattern}
Let $f:V(\G) \to \C$ be a harmonic embedding of a planar network $\G$ in a convex polygon $P$;
let $g: V(\G^{*}) \to \C$ be its dual. We define a realization $z:V(\GD) \to \mathbb{C}$ of the bipartite graph $\GD$ such that $z = f$ for the black vertices coming from the vertices of $\G$ and $z= g$ for those from the faces of $\G$. On the white vertices, we take $z$ as the intersection of the line through the primal edge and the line through the dual edge under $f$ and $g$.
Then $z$ has cyclic faces and thus is a circle pattern with the combinatorics of $\GD$. 
The face weights induced on $\GD$ from the circle pattern coincide with those from Temperley's bijection.
\end{theorem}
\begin{proof}
Since every dual edge of $\G$ is perpendicular to its primal edge under the harmonic embeddings, 
the quadrilateral faces of $\GD$ have right angles at their white vertices. 
Hence every face of $z$ is cyclic and hence we obtain a circle pattern. 
The circumcenter of each cyclic face of $z$ is the midpoint of the two black vertices. 
By similarity of triangles, the edge weight induced from the distance between circumcenters has the following form: For an edge of $\GD$ that is a half-edge of a primal edge $e$ of $\G$, it has weight $\ell_{e*}/2$. For an edge of $\GD$ that is a half-edge of a dual edge $e^*$, it has weight $\ell_{e}/2$. Thus for every quadrilateral face $\phi$, the face weight is
\[
X_{\phi} = \frac{\ell_{e_1^*}}{\ell_{e_1}} \frac{\ell_{e_2}}{\ell_{e_2^*}} = \frac{c_{e_1}}{c_{e_2}}
\]
which coincides with that from Temperley's bijection.
\end{proof}

\subsection{Star-triangle relation}

It is a well-known fact \cite{Epifanov1966} that a network can be reduced to the trivial network by performing star-triangle and triangle-star moves, as well as two other types of moves: replacing two parallel edges (sharing the same endpoints)
with a single edge, and replacing two edges in series with a single edge (that is, deleting a degree-$2$ vertex).

The star-triangle move has a simple interpretation in terms of reciprocal figures: it corresponds exactly to Steiner's theorem (see Figure~\ref{fig:Steiner}), as was observed in \cite{Schief2002a}. The star-triangle move corresponds to replacing a vertex which is the intersection of three primal edges (such as $D'$ on Figure~\ref{fig:Steiner}) by a dual vertex which is the intersection of three dual edges (such as $D$ on Figure~\ref{fig:Steiner}); 
Steiner's theorem guarantees that these three dual edges intersect at a common point.

\begin{figure}
\begin{center}
\includegraphics[width=4in]{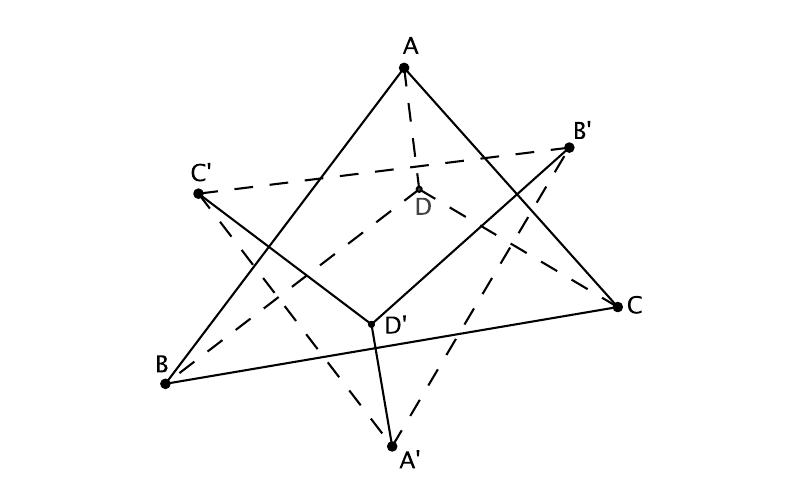}
\end{center}

\caption{Steiner's theorem (see e.g. \cite[Figure 4.9.18]{Akopyan2011}) states that the perpendicular to $(AB)$ going through $C'$, the perpendicular to $(BC)$ going through $A'$ and the perpendicular to $(AC)$ going through $B'$ are concurrent if and only if the perpendicular to $(A'B')$ going through $C$, the perpendicular to $(B'C')$ going through $A$ and the perpendicular to $(A'C')$ going through $B$ are concurrent.}\label{fig:Steiner}
\end{figure}

In \cite{GK2013} it was observed that a $Y-\Delta$ transformation for planar networks can be decomposed into a composition of four urban renewals for dimer models, upon transforming the planar network into a dimer model via Temperley's bijection. We show that this decomposition can be seen in purely geometric terms, using the correspondences between planar networks and reciprocal figures on the one hand, and between dimer models and circle patterns on the other hand.

\begin{theorem}
The star-triangle move for reciprocal figures can be decomposed into four Miquel moves, upon transforming the reciprocal figures into a circle pattern as described in Theorem~\ref{thm:powerdiagramtocirclepattern}.
\end{theorem}

\begin{figure}
\begin{center}
\includegraphics[width=6.0in]{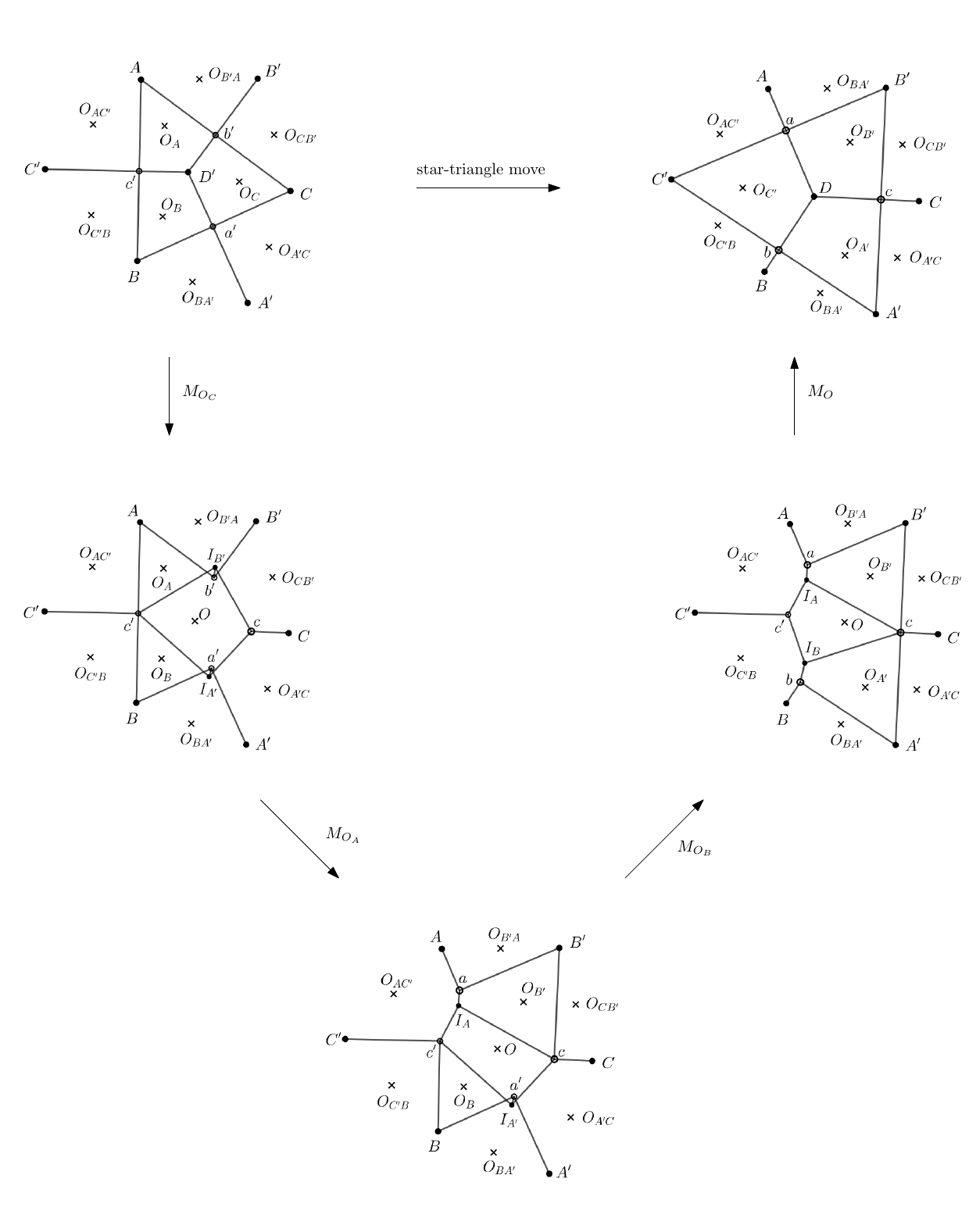}
\end{center}
\caption{\label{fig:SteinerMiquel}Decomposition of a star-triangle move for reciprocal figures into four Miquel moves.}
\end{figure}

\begin{proof}

This decomposition is illustrated in Figure~\ref{fig:SteinerMiquel}. We start with a triangle $ABC$ in a harmonic embedding, we denote by $D'$ the dual vertex associated with that triangle and by $A',B'$ and $C'$ the three dual vertices adjacent to $D'$. We construct the circle pattern associated with the reciprocal figures as described in Theorem~\ref{thm:powerdiagramtocirclepattern}, denoting by $a',b'$ and $c'$ the intersections of the primal edges and their associated dual edges. We respectively denote by $O_A$, $O_B$ and $O_C$ the centers of the circumcircles of the quadrilaterals $Ac'D'b'$, $Ba'D'c'$ and $Cb'D'a'$. We also respectively denote by $O_{AC'}$, $O_{C'B}$, $O_{BA'}$, $O_{A'C}$ $O_{CB'}$ and $O_{B'A}$ the circumcenters of the triangles $Ac'C'$, $C'c'B$, $Ba'A'$, $A'a'C$, $Cb'B'$ and $B'b'A$. We first apply the Miquel move $M_{O_C}$ to the quadrilateral $D'a'Cb'$ with circumcenter $O_C$. The points $D'$, $a'$, $C$ and $b'$ respectively transform into $c'$, $I_{A'}$, $c$ and $I_{B'}$, which form a cyclic quadrilateral with circumcenter denoted by $O$. Then we apply the Miquel move $M_{O_A}$ to the quadrilateral $Ac'I_{B'}b'$ with circumcenter $O_A$. The points $A$, $c'$, $I_{B'}$ and $b'$ respectively transform into $a$, $I_{A}$, $c$ and $B'$, which form a cyclic quadrilateral with circumcenter denoted by $O_{B'}$. Next we apply the Miquel move $M_{O_B}$ to the quadrilateral $Ba'I_{A'}c'$ with circumcenter $O_B$. The points $B$, $a'$, $I_{A'}$ and $c'$ respectively transform into $b$, $A'$, $c$ and $I_B$, which form a cyclic quadrilateral with circumcenter denoted by $O_{A'}$. Finally we apply the Miquel move $M_{O}$ to the quadrilateral $I_Ac'I_Bc$ with circumcenter $O$. The points $I_A$, $c'$, $I_B$ and $c$ respectively transform into $a$, $C'$, $b$ and $D$, which form a cyclic quadrilateral with circumcenter denoted by $O_{C'}$.

We now show that this point $D$ created by a composition of four Miquel moves coincides with the point $\tilde{D}$ created by the star-triangle move applied to the reciprocal figures. First, as observed in the proof of Theorem~\ref{thm:powerdiagramtocirclepattern}, in a circle pattern coming from reciprocal figures, the center of each circle is the midpoint of the segment formed by the two black vertices. Since $O_{AC'}$ is the circumcenter of the triangle $AC'a$ and is the midpoint of $[AC']$, this implies that the perpendicular to $(B'C')$ going through $A$ is the line $(Aa)$. Similarly, $(Bb)$ is the perpendicular to $(A'C')$ going through $B$ and $(Cc)$ is the perpendicular to $(A'B')$ going through $C$. Hence the point $\tilde{D}$ created by the star-triangle move is the intersection point of the three lines $(Aa)$, $(Bb)$ and $(Cc)$. Because of the orthogonality property at $a,b$ and $c$, the point $\tilde{D}$ lies on the circumcircles of the three triangles $aC'b$, $bA'c$ and $cB'a$ so $\tilde{D}=D$.
\end{proof}

\section{From Ising s-embeddings to circle patterns}
\label{sec:sembeddings}

We consider the Ising model on a planar graph $\G$ with edge weights $x_e$. 
Chelkak introduced in~\cite{Chelkak2017} an \textbf{s-embedding} of $\G$, which is an embedding $s$ defined on each vertex, dual vertex and vertices associated to edges of $\G$ with the following property: for any edge $e$ in $\G$, if $v_0^\bullet$ and $v_1^\bullet$ (resp. $v_0^\circ$ and $v_1^\circ$) denote the endpoints of $e$ (resp. of the edge dual to $e$) and $v_e$ a vertex associated to $e$ as on Figure~\ref{fig:Isingtodimer}, then $s(v_0^\bullet)$, $s(v_0^\circ)$, $s(v_1^\bullet)$ and $s(v_1^\circ)$ form a tangential quadrilateral with incenter $s(v_e)$, meaning that there exists a circle centered at $s(v_e)$ and tangential to the four sides of the quadrilateral.

\begin{figure}[h!]
\center{\includegraphics[width=6in]{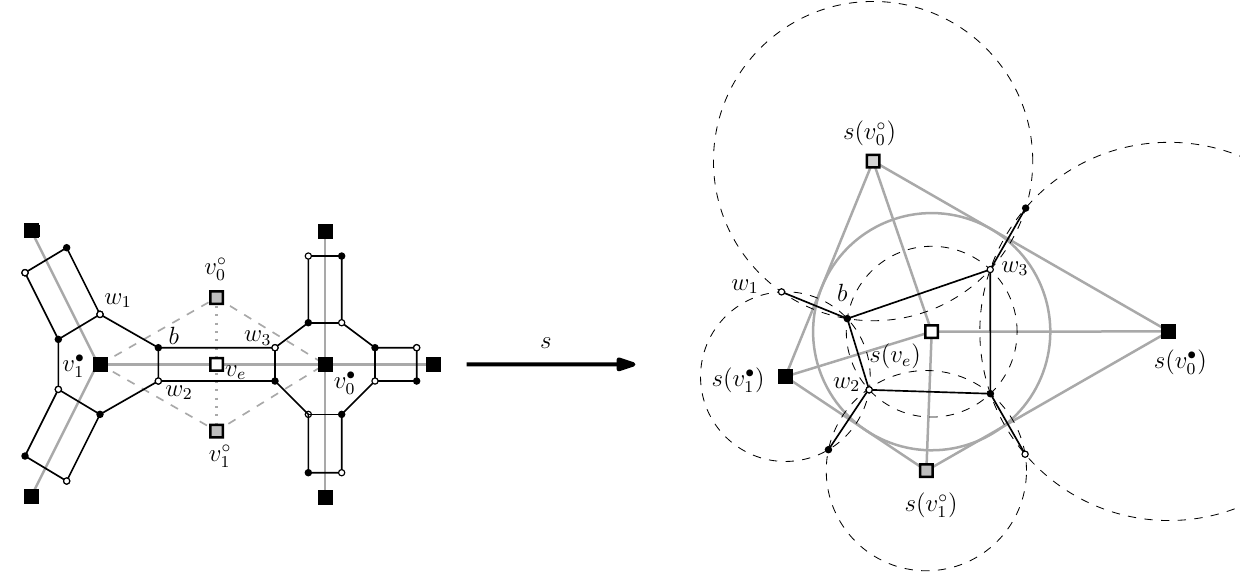}}
\caption{Ising graph $\G$ (black rhombi); dual graph $\G^*$ (gray rhombi); dimer graph $\G_D$ (black and white vertices).}\label{fig:Isingtodimer}
\end{figure}

On the other hand, Dub\'edat \cite{Dubedat2011} gave a natural map from the Ising model on $\G$ to a bipartite 
dimer model $\G_D$, as in Figure~\ref{fig:Isingtodimer}: Each edge in $\G$ is replaced by a quadrilateral in $\G_D$ and each vertex or face of degree $d$ in $\G$ is replaced by a face of degree $2d$ in $\G_D$. Every face of $\G_D$ corresponds to a vertex, an edge or a face of $\G$. For every edge $e$ of $\G$, define 
$\theta_e\in(0,\pi)$ by
\[
x_e=\tan\frac{\theta_e}{2}.
\]
Then we define the edge weights on $\G_D$ by the following formulas (adopting the notation of Figure~\ref{fig:Isingtodimer}):
\[
\omega(b,w_1)=1, \quad \omega(b,w_2)=\cos\theta_e, \quad \omega(b,w_3)=\sin\theta_e.
\]
For these weights the partition function of the Ising model on $\G$ is equal (up to a multiplicative constant) to the partition function of the dimer model on $\G_D$, see \cite{Dubedat2011}.

The goal of this section is to show that the following diagram commutes:
$$\begin{array}{ccc}\text{Ising model on $\G$} &\longrightarrow& \text{Bipartite graph $\G_D$} \\
\updownarrow&&\updownarrow\\
\text{s-embedding} &\longrightarrow& \text{Circle pattern} 
\end{array}
$$

In particular, the s-embedding of the vertices, dual vertices and edge midpoints of $\G$ coincide with the circle centers associated with the bipartite graph $\G_D$. Note that combinatorially this is consistent since each face of $\G_D$ corresponds to either a vertex, a face or an edge of $\G$.

\begin{theorem}
An $s$-embedding of $\G$ provides an embedding of $\G^*_D$ into $\mathbb{C}$ sending each vertex of $\G^*_D$ to the centers of a circle pattern associated with $\G_D$.
\end{theorem}

\begin{proof}
It suffices to prove that, for each face of the bipartite graph $\G_D$, the alternating product of the edge weights $\omega$ induced by $s$ satisfies (\ref{eq:realcenter}), where $X$ is the face weight of~$\G_D$. 

First, we check the conditions on the faces of $\G_D$ that correspond to vertices or faces of $\G$. By symmetry, it suffices to consider just a face of $\G$. Let $v^*$ be a vertex of the dual graph of $\G_D$ which corresponds to a face of $\G$ of degree $d$ and denote by $v_{e_1},v_1,v_{e_2},v_2,\ldots,v_{e_d},v_d$ the neighbors of $v^*$ in $\G_D^*$ in counterclockwise order, where the vertices of type $v_{e_i}$ correspond to an edge in $\G$ while the vertices of type $v_i$ correspond to a vertex in $\G$. The weight of an edge in $\G_D$ dual to an edge of type $v^* v_{e_i}$ (resp. $v^* v_i$) is of the form $\sin \theta_i$ (resp. is equal to $1$). Hence we need to show the following two formulas:
\[
\arg\frac{\prod_{i=1}^d s(v_{e_i})- s(v^*)}{\prod_{i=1}^d s(v_i)- s(v^*)} = \pi \quad \text{and} \quad
\frac{\prod_{i=1}^d |s(v_{e_i})- s(v^*)|}{\prod_{i=1}^d |s(v_i)- s(v^*)|} = \prod_{i=1}^d \sin \theta_i.
\]
By splitting each formula into $d$ equations centered around the edges of type $v^* v_{e_i}$, it suffices to prove the following two formulas, where we are using the notation of Figure~\ref{fig:Isingtodimer}:
\begin{align}
\arg \frac{s(v_e)-s(v_0^\circ)}{s(v_0^\bullet)-s(v_0^\circ)} &=\arg \frac{s(v_1^\bullet)-s(v_0^\circ)}{s(v_e)-s(v_0^\circ)}, \label{eq:faceargument} \\
\sin^2\theta_e &=\frac{|s(v_e)-s(v_0^\circ)|^2}{|s(v_1^\bullet)-s(v_0^\circ)|\cdot|s(v_0^\bullet)-s(v_0^\circ)|}. \label{eq:faceXvariable}
\end{align}

Formula~\eqref{eq:faceargument} follows from the fact that $s(v_e)$ is the center of the incircle of the quadrilateral with vertices $s(v_0^\bullet)$, $s(v_0^\circ)$, $s(v_1^\bullet)$ and $s(v_1^\circ)$. For the other formula, we start from~\cite[Formula (6.3)]{Chelkak2017} which implies that
\[
\tan^2 \theta_e = \frac{|s(v_0^\circ)-s(v_e)|\cdot|s(v_1^\circ)-s(v_e)|}{|s(v_0^\bullet)-s(v_e)|\cdot|s(v_1^\bullet)-s(v_e)|},
\]
hence
\begin{equation}
\label{eq:Dimarewritten}
\frac{1}{\sin^2 \theta_e} = \frac{|s(v_0^\circ)-s(v_e)|\cdot|s(v_1^\circ)-s(v_e)|+|s(v_0^\bullet)-s(v_e)|\cdot|s(v_1^\bullet)-s(v_e)|}{|s(v_0^\circ)-s(v_e)|\cdot|s(v_1^\circ)-s(v_e)|}.
\end{equation}
Furthermore, we have the following classical formula for tangential quadrilaterals
\begin{multline}
\label{eq:propertyoftangentialquads}
|s(v_1^\bullet)-s(v_0^\circ)|\cdot|s(v_0^\bullet)-s(v_0^\circ)| = \\
 \frac{|s(v_0^\circ)-s(v_e)|}{|s(v_1^\circ)-s(v_e)|}\left(|s(v_0^\circ)-s(v_e)|\cdot|s(v_1^\circ)-s(v_e)|+|s(v_0^\bullet)-s(v_e)|\cdot|s(v_1^\bullet)-s(v_e)|\right).
\end{multline}
To see this formula, suppose that the center $s(v_e)$ of the inscribed circle is at the origin of the complex plane. Let $a,b,c,d \in\mathbb{C}$ be the respective positions of vertices $s(v_1^\bullet),s(v_1^\circ),s(v_0^\bullet),s(v_0^\circ)$. Note that $\frac{(a-b)(c-b)}{b^2} \in \mathbb{R}_{+}$ and $\frac{ac}{bd}\in\mathbb{R}_{-}$ since $s(v_1^\bullet)s(v_1^\circ)s(v_0^\bullet)s(v_0^\circ)$ is a tangential quadrilateral. Note that $1/\bar a$ is the midpoint of the segment connecting the tangency points where the incircle touches the sides $s(v_1^\bullet)s(v_1^\circ)$ and $s(v_1^\bullet)s(v_0^\circ)$, and similarly
for $1/\bar b,1/\bar c,1/\bar d$. Thus $1/\bar a+1/\bar c=1/\bar b+1/\bar d$: 
both sides of this equation are twice the 
geocenter of the four tangency points. 
Conjugating, this is equivalent to the expression $\frac{(a-b)(c-b)}{b^2}=1-\frac{ac}{bd}$, which is in turn equivalent to~(\ref{eq:propertyoftangentialquads}).

Next, we check these conditions for faces of $\G_D$ corresponding to edges of $\G$. We need to show the following two formulas:
\begin{align}
\arg \frac{(s(v_0^\bullet)-s(v_e))(s(v_1^\bullet)-s(v_e))}{(s(v_0^\circ)-s(v_e))(s(v_1^\circ)-s(v_e))} &=\pi, \label{eq:edgeargument} \\
\frac{|s(v_0^\bullet)-s(v_e)|\cdot|s(v_1^\bullet)-s(v_e)|}{|s(v_0^\circ)-s(v_e)|\cdot|s(v_1^\circ)-s(v_e)|}&=\frac{\cos^2\theta_e}{\sin^2\theta_e}. \label{eq:edgeXvariable}
\end{align}

Formula~\eqref{eq:edgeargument} follows from subdividing the quadrilateral with vertices $s(v_0^\bullet)$, $s(v_0^\circ)$, $s(v_1^\bullet)$ and $s(v_1^\circ)$ into four triangles sharing the common vertex $s(v_e)$, taking the alternating sum of four formulas of the type of~\eqref{eq:faceargument}. Formula~\eqref{eq:edgeXvariable} follows immediately from formula (6.3) in~\cite{Chelkak2017}.
\end{proof}

\section{Appendix}

Let $M=(m_{ij})_{i,j = 1,\dots, n}$ be an invertible matrix such that $m_{j1}^{-1}\neq 0$ for each $j$. 
Can we find diagonal matrices $F$ and $G$ so that $FMG$ has given row and column sums
(subject to the obvious condition that the sum of the row sums equals the sum of the column sums)?
If $M$ has positive entries this follows from Sinkhorn's Theorem \cite{sinkhorn1964}.
We show here that for $M$ of the above form it is true for a Zariski dense set of row and column sums.

Consider the map $\mathfrak{F}: \mathbb C^{2n}\to \mathbb C^{2n}$ given coordinate-wise by
\begin{equation*}
  \begin{split}
    & p_i(f_1,\dots, f_n, g_1,\dots, g_n) = \sum_{j = 1}^n f_im_{ij} g_j,\\
    & q_j(f_1,\dots, f_n, g_1,\dots, g_n) = \sum_{i = 1}^n f_im_{ij} g_j,
  \end{split}
\end{equation*}
where $p_i,q_j$ are coordinates on the terminal $\mathbb C^{2n}$ and $f_i, g_j$ are coordinates on the initial~$\mathbb C^{2n}$. It is clear that the image of $\mathfrak{F}$ is contained in the hyperplane
\[
  \Sigma = \left\{ (p_1,\dots, p_n, q_1,\dots, q_n)\in \mathbb C^{2n}\ \mid\ \sum_{i = 1}^n p_i = \sum_{i = 1}^n q_i \right\}.
\]

\begin{lemma}\label{Zariski}
The image of $\mathfrak{F}$ is Zariski dense in $\Sigma$.
\end{lemma}

\begin{proof}
It suffices to prove that the Jacobian of $\mathfrak F$ is of (maximal) rank $2n-1$ at some point. In order to do this we first restrict $\mathfrak F$ to the subset $(\mathbb C^*)^n\times \mathbb C^n$ and then write it as a composition ${\Phi}\circ \Psi$, where
\begin{equation*}
\begin{split}
&\Psi(f,g) = (p,f),\quad \text{where } p_i = \sum_{j = 1}^n f_im_{ij} g_j;\\
&\Phi(p, f) = (p,q),\quad \text{where } q_j = \sum_{l = 1}^n f_l m_{lj}\sum_{i = 1}^n m^{-1}_{ji}\frac{p_i}{f_i}.
\end{split}
\end{equation*}
Note that the map $\Psi$ is invertible since given $p$ and $f$ one can reconstruct $g$. Indeed, $g_j = \sum_{i = 1}^n m^{-1}_{ji} \frac{p_i}{f_i}$. Since $\Psi$ is invertible it is enough to find a point where the Jacobian of $\Phi$ has the maximal rank. Fix variables $p$ and consider $\Phi$ as the map from the space with coordinates $f$ to the space with coordinates $q$. Let $p_1 = 1$ and $p_2= \dots = p_n = 0$. For this particular choice of $p$ we get
 \[
    \Phi(p, f) = (p, q), \quad q_j = m^{-1}_{j1}\frac{p_1}{f_1} \sum_{l = 1}^n f_l m_{lj}.
  \]
  Note that the right-hand side of the equation for $q$ is linear in $f_2,\dots, f_n$. Since the matrix $(m_{ij})$ is invertible we conclude that the Jacobian of the map that sends $f$ to $q$ is of rank $n-1$. Hence the Jacobian of $\Phi$ is of rank $2n-1$.
\end{proof}

\section*{Acknowledgements}

We are grateful to Dmitry Chelkak for sharing his ideas during many fruitful exchanges, 
and providing the proof of Lemma \ref{lem:doublepoint}. We also thank the referees
for helpful comments and suggestions.
R. Kenyon is supported by NSF grant DMS-1713033 and the Simons Foundation award 327929. S. Ramassamy acknowledges the support of the Fondation Simone et Cino Del Duca, of the ENS-MHI chair funded by MHI, of the Fondation Sciences Math\'ematiques de Paris and of the ANR-18-CE40-0033 project DIMERS. M. Russkikh is supported by NCCR SwissMAP of the SNSF and ERC AG COMPASP. 

\bibliographystyle{siam}
\bibliography{circles}

\end{document}